\documentclass[preprint,12pt]{elsarticle}
\usepackage[T1]{fontenc} 
\usepackage{amssymb} 
\usepackage{amsthm}  
\usepackage{amsmath}
\usepackage{amsfonts}
\usepackage{mathtools}
\usepackage{verbatim}
\usepackage{color,soul}
\usepackage{graphicx}
\graphicspath{ {images/} } 
\usepackage{subfigure}
\usepackage[section]{placeins} 
\usepackage{lineno}  
\usepackage[colorlinks,linkcolor=red,anchorcolor=red,citecolor=red]{hyperref}
\usepackage{cleveref}
\usepackage[dvipsnames]{xcolor}
\newtheorem{pro}{Proposition}

\newtheorem{rem}{Remark}

\newtheorem{lem}{Lemma}
\newtheorem{thm}{Theorem}

\numberwithin{equation}{section}
\numberwithin{thm}{section}
\numberwithin{lem}{section}
\numberwithin{rem}{section}

\newdefinition{rmk}{Remark}[section]

\renewenvironment{proof}{{\bfseries Proof}}

\journal{Submitted to XXX Journal}
\begin{document}

\begin{highlights}	
	\item The first study of energy-stable ODE discretization and the energy-stable fully discrete scheme by the SPH method for two-phase problems 
	\item The scheme ensures the inheritance of momentum conservation and the energy dissipation law from the PDE level to the ODE level, and then to the fully discrete level. 
	\item This novel method based on NSCH model enjoys the physical consistency, and the detailed mathematical proof is also provided. 
	\item The scheme helps increase the stability of the numerical method, which allows much larger time step sizes than the traditional ISPH method. 
	\item This energy-stable scheme not only alleviates tensile instability, but it also captures the behavior of the interface and energy variation well.
\end{highlights}

\begin{frontmatter}
\title{An energy-stable Smoothed Particle Hydrodynamics discretization of the Navier-Stokes-Cahn-Hilliard model for incompressible two-phase flows} \tnotemark[1]
\tnotetext[1]{https://ctpl.kaust.edu.sa}
\author[inst1]{Xiaoyu Feng} 
\affiliation[inst1]{organization={Computational Transport Phenomena Laboratory (CTPL), Division of Physical Sciences and Engineering (PSE), King Abdullah University of Science and Technology (KAUST)},
            addressline={Thuwal}, 
            postcode={23955-6900}, 
            country={Kingdom of Saudi Arabia}}
\ead{xiaoyu.feng@kaust.edu.sa}          
\author[inst2]{Zhonghua Qiao}
\affiliation[inst2]{organization={Department of Applied Mathematics and Research Institute for Smart Energy, The Hong Kong Polytechnic University},
            addressline={ Hung Hom, Kowloon, Hong Kong},
            city={Hong Kong}}
\ead{zhonghua.qiao@polyu.edu.hk}            
\author[inst1]{Shuyu Sun\corref{cor1}} 
\ead{shuyu.sun@kaust.edu.sa}

\author[inst3]{Xiuping Wang}
\affiliation[inst3]{organization={Computational Transport Phenomena Laboratory (CTPL), Division of Computer, Electrical and Mathematical Sciences and Engineering (CEMSE), King Abdullah University of Science and Technology (KAUST)},
            addressline={Thuwal}, 
            postcode={23955-6900}, 
            country={Kingdom of Saudi Arabia}}
\ead{xiuping.wang@kaust.edu.sa}
\cortext[cor1]{Corresponding author}

\begin{abstract}
Varieties of energy-stable numerical methods have been developed for incompressible two-phase flows based on the Navier-Stokes–Cahn–Hilliard (NSCH) model in the Eulerian framework, while few investigations have been made in the Lagrangian framework. Smoothed particle hydrodynamics (SPH) is a popular mesh-free Lagrangian method for solving complex fluid flows. In this paper, we present a pioneering study on the energy-stable SPH discretization of the NSCH model for incompressible two-phase flows. We prove that this SPH method inherits mass and momentum conservation and the energy dissipation properties at the fully discrete level. With the projection procedure to decouple the momentum and continuity equations, the numerical scheme meets the divergence-free condition.  Some numerical experiments are carried out to show the performance of the proposed energy-stable SPH method for solving the two-phase NSCH model.  The inheritance of mass and momentum conservation and the energy dissipation properties are verified numerically.
\end{abstract}

\begin{keyword}
Smoothed Particle Hydrodynamics \sep Energy stability \sep Two-phase Flow \sep Navier-Stokes-Cahn-Hilliard model
\end{keyword}

\end{frontmatter}



\section{Introduction}  \label{section:Introduction}
Initially proposed \cite{gingold1977smoothed,lucy1977numerical} for astrophysical problems, the smoothed particle hydrodynamics (SPH) approach has become one of the most popular mesh-free particle methods and it is now widely used in various applications, such as multi-phase flow \cite{hu2009constant}, flow in porous media \cite{tartakovsky2016smoothed}, fracture and crack propagation \cite{chen2021modified}, lava flow \cite{zago2018semi}, magneto-hydrodynamics \cite{price2012smoothed}, the evolution of planets and stars \cite{nakajima2015melting}, and even movie special effects \cite{ihmsen2014sph}. Unlike the nodes in other mesh-free methods \cite{wang2002point}, the SPH particles can be regarded not only mathematically as interpolation points but also physically as material components, just like the ``embedded atoms" in molecular dynamics (MD) simulation \cite{hoover1998isomorphism}.

SPH possesses many well-known appealing features. 
First of all, compared to well-studied grid-based Eulerian numerical methods, SPH naturally takes advantage of the Lagrangian framework, and it treats the convection term effectively and stably.  
Free surfaces, material interfaces, and moving boundaries can all be traced naturally through this method without tracking or reconstructing. This mesh-free feature facilitates applications like multi-phase flow and high-energy events like explosions, high-speed collisions, and penetrations. 
Local mass conservation is automatically retained in SPH since the mass is carried by each particle as a unique property of the particle.
By using the unique ``color'' property of each particle, it is much easier to describe complex multi-component or multi-material phenomena, such as the hydrocarbon components in oil and gas reservoirs. 
In addition, the SPH method can handle large deformations and complicated geometry without causing mesh distortion. 
With the progress of neighbor searching schemes and parallel computing strategies \cite{nishiura2015computational}, SPH enjoys even greater computational performance and potential. 
Some review literature for the SPH method can be referred to \cite{monaghan2005smoothed,liu2010smoothed}. 

In this paper, we consider incompressible two-phase flows based on the Navier-Stokes–Cahn–Hilliard (NSCH) model. Even though possessing a lot of the appealing features mentioned above, SPH still faces a number of open problems and challenges. First of all, the energy-stability and physical consistency of SPH for the NS equation have not yet been rigorously studied. Second, the treatment of incompressibility in SPH remains tricky and challenging; in fact, most incompressibility treatments in the literature are not energy stable. Third, to the best of our knowledge, there is not an energy-stable SPH method for the NSCH system proposed in the literature. 

Two common perspectives for the stability analysis of the SPH method include: 
(1) tensile instability caused by disordered distributions of particles; (2) the maximum time step that keeps the simulation stable. For the former, a novel technique called ``particle shifting'' has been employed to alleviate particle gathering \cite{khayyer2017comparative}. For the second, to date, several attempts have been made to study the energy variation of the SPH methods, which reveal that the entropy-increasing (energy dissipation) methods possess better stability \cite{price2012smoothed}. In many approaches, the energy dissipation was explicitly added by an artificial viscosity term, but the artificial viscosity term can introduce additional numerical errors (commonly known as numerical diffusion). 
For the treatment of the incompressibility of the fluid, two common approaches are used: (1) approximately simulating incompressible flows with small compressibility, known as Weakly Compressible SPH (WCSPH); (2) simulating incompressible flows by enforcing the incompressibility, known as Incompressible SPH (ISPH)\cite{xu2009accuracy}. However, energy stable treatment of incompressibility has not fully been addressed in the literature. 
Recently, Sun and Zhu \cite{zhu2022energy} proposed an energy-stable SPH method for incompressible single-phase flow that does not use artificial viscosity terms, and it serves as the basis and inspiration for this paper.

The interface, as a critical element of the two-phase fluid system, can be modeled either as a sharp interface   \cite{osher1988fronts,hirt1981volume,sun2010coupled} or as a diffuse interface \cite{gibbs1878equilibrium}. 
The NSCH model belongs to the diffuse interface approach, and it obeys thermodynamically consistent energy dissipation laws. The same energy law is also desired to be retained in the discretized equations. 
Some techniques, like convex-concave splitting \cite{eyre1998unconditionally,feng2018novel} and the stabilizing approach, are used to construct energy-stable schemes. Several advances in energy-stable schemes include the scalar auxiliary variable method (SAV), the invariant energy quadratization method (IEQ) \cite{shen2018scalar}, the exponential time differencing method (ETD) \cite{du2019maximum}, and the linear energy-factorization method (EF) \cite{wang2021linear,feng2022fully}. The diffuse interface model based on the EoS is also proposed for modeling complex two-phase and multi-component mixtures \cite{qiao2014two,kou2020novel,fan2017componentwise}. However, the above studies are all based on a mesh and the Euler framework. Most of the SPH two-phase methods are based on the sharp interface model \cite{yang2019comprehensive}. Researchers rarely looked into how to combine the SPH method and the diffuse interface model together \cite{hirschler2014application} and how to design an energy-stable scheme for the system. 

The rest of this paper is organized as follows. In Section \ref{section:PDE}, we review the NSCH system and give a brief proof of energy law at the PDE level. In Section \ref{section:SPH&ODE}, some basics of the SPH method are given and an energy-stable ODE model is proposed based on the SPH framework. In Section \ref{section:Scheme}, an energy-stable fully discrete scheme is well developed. We provide the detailed proof and derivation of the ODE system and fully discrete scheme. In Section \ref{section:Numerical examples}, three numerical examples are presented to validate the scheme. Finally, some concluding remarks are given in Section \ref{section:Conclusion}.

\section{PDE model and its energy law}  \label{section:PDE}
\setcounter{equation}{0}
\renewcommand{\theequation}{\arabic{section}.\arabic{equation}}
We now study a mixture of two immiscible, incompressible fluids in a confined domain $\Omega \subset \mathbb{R}^{d}(d=2,3)$ and the NSCH model is introduced. To identify the regions occupied by the two fluids, we introduce a phase function $\phi$, such that 
\begin{equation*}
\phi(\mathbf{x}, t)= \begin{cases}1 & \text { fluid 1 } \\ -1 & \text { fluid 2 }\end{cases}
\end{equation*}
According to the idea of the diffuse interface and gradient flow theory, one of the most popular thermodynamic theories for inhomogeneous fluids, the total mixing energy has two contributions: $F_{b}(\phi)$ from the homogeneous part of the fluid and $F_{\nabla}(\phi)$ from the inhomogeneity of the fluid. The thermodynamic behavior of the entire two-phase system is governed by the mixing energy functional $F(\phi, \nabla \phi)$,
\begin{equation}  \label{TotalE}
\begin{split}
F(\phi, \nabla \phi) & = \int_{\Omega} f(\phi, \nabla \phi) d \mathbf{x} = F_{b}(\phi) + F_{\nabla}(\phi)  = \int_{\Omega} f_{b}(\phi) d \mathbf{x}+\int_{\Omega} f_{\nabla}(\phi) d \mathbf{x}\\
 & = \int_{\Omega}\left\{f_{b}(\phi) + \frac{\lambda}{2}\|\nabla \phi\|^{2} \right\} d \mathbf{x},
\end{split}
\end{equation}
where $f$ represents the energy density, and $\lambda$ denotes the characteristic strength of the phase mixing energy with respect to $\phi$. $\lambda$ has a relation with the surface tension coefficient $\sigma$ at the equilibrium state: $\lambda=\frac{3 \sigma}{2 \sqrt{2}} \varepsilon$, where $\varepsilon$ is the capillary width of the interface thickness. The energy density function from homogeneity, $f_{b}(\phi)$, only depends on $\phi$ locally and $f_b(\phi)=\lambda\left(\phi^{2}-1\right)^{2} /\left(4 \varepsilon^{2}\right)$. The equilibrium process of a two-phase system minimizes the total mixing energy \eqref{TotalE}. Through the variational derivative of the energy functional $F$ with respect to $\phi$, we obtain the chemical potential $\mu$ and $\mu=\mu_{b}+\mu_{\nabla}$. By setting the coefficient $\lambda_{b}=\lambda/\varepsilon^{2}$, we have
$\displaystyle\mu_{b}: = \frac{\delta F_{b}}{\delta \phi} =  f_b^{\prime}\left(\phi\right) = \lambda_{b}(\phi^3-\phi)$, and
 $\displaystyle\mu_\nabla:= \frac{\delta F_{\nabla}}{\delta \phi} = -\lambda \Delta \phi.$

The dynamics of the phase function $\phi$ can be determined by an $H^{-1}$ gradient
flow of \eqref{TotalE}, which leads to the following Cahn-Hilliard (CH) equations:
\begin{subequations}
\begin{eqnarray}
\frac{D \phi}{D t}= \frac{\partial \phi}{\partial t}+(\mathbf{v} \cdot \nabla) \phi = \nabla \cdot(M \nabla \mu),  \label{CHpde1}
\\
\mu=\lambda_{b}(\phi^3-\phi) - \lambda \Delta \phi.	\label{CHpde2}
\end{eqnarray}
\end{subequations}
Here, $M$ is a mobility parameter related to the relaxation time scale, and $\mathbf{v}$ is the velocity of the fluids. When it turns to the Navier-Stokes (NS) part, for the two-phase system, the momentum equation takes the usual form below.
\begin{equation*}
\rho \frac{D \mathbf{v}}{D t} = \rho\left( \frac{\partial \mathbf{v}}{\partial t} +(\mathbf{v} \cdot \nabla) \mathbf{v}\right) = -\nabla p+\nabla \cdot \left( \eta  D(\mathbf{v})\right) -\lambda \nabla \cdot( \nabla \phi \otimes \nabla \phi),
\end{equation*}
where $\rho$ denotes the density, $p$ is the pressure, and the shear viscosity is denoted by $\eta$. $-\lambda \nabla \phi \otimes \nabla \phi$ represents the extra elastic stress induced by the interfacial tension. With the aid of modified pressure and the divergence theorem, we arrive at the mathematical formula for the matched density and viscosity case:
\begin{equation}
\rho \frac{D \mathbf{v}}{D t}=-\nabla p+\eta \Delta \mathbf{v}+\mu \nabla \phi.  \label{NSpde1}
\end{equation}
with the continuity equation characterizing the mass conservation:
\begin{equation}\label{ContinuityE}
\frac{D \rho }{D t} = \frac{\partial \rho}{\partial t}+\nabla \cdot(\rho \mathbf{v})=0.
\end{equation}
For incompressible flows, \eqref{ContinuityE} can be converted into the divergence-free condition:
\begin{equation}
\nabla \cdot \mathbf{v}=0. \label{div-free}	
\end{equation}
The periodic boundary condition is applied here:
\begin{equation}
\Phi(x)=\Phi(x+nL),  \quad n \in Z, \label{per-bdry}
\end{equation}
where $\Phi$ denotes a general variable, and it can be $\phi$, $p$, $\mathbf{v}$, $\rho$ and so on. $L$ is a period of function. This periodic boundary condition makes $\oint_{\partial \Omega} \Phi(\mathbf{x}) \cdot \mathbf{n} d s=0$ hold.

We now derive the energy law at the PDE level for the NSCH system includes \eqref{CHpde1} , \eqref{CHpde2}, \eqref{NSpde1}, \eqref{div-free} and \eqref{per-bdry}. Here and after, for any functions $f, g \in L^{2} (\Omega)$, we use $(f, g)$ to denote $\int_{\Omega} f g \mathrm{~d} \mathbf{x}$ and $\|f\|^{2}=(f,f)$.
\begin{pro}
The NSCH system satisfies the following energy dissipation law:
\begin{equation}
\frac{d E_{\text {total }}}{d t}
= \frac{d E_{k}}{d t}+\frac{d F}{d t} = - \eta\|\nabla \mathbf{v}\|^{2} - M\|\nabla \mu\|^{2}  \leq 0,   \label{NSCHPDEenergy}
\end{equation}
where the kinetic energy $\displaystyle E_{k}:=\frac{\rho}{2}\|\mathbf{v}\|^{2}$ and the total energy $E_{\text {total }}=E_k+F$.
\end{pro}
\begin{proof}
By taking the inner product of equation \eqref{NSpde1} with $\mathbf{v}$ and applying the divergence-free condition \eqref{div-free}, we get
\begin{equation*}
 \frac{d E_{k}}{d t} = \frac{d(\frac{1}{2} \rho\|\mathbf{v}\|^{2})}{d t}
 = ( \rho \frac{D \mathbf{v}}{D t}, \mathbf{v} )  =  (-\nabla p, \mathbf{v}) +
 ( \eta \Delta \mathbf{v}, \mathbf{v} ) + ( \mu \nabla \phi , \mathbf{v} ).
\end{equation*}
Combining the inner product of equation \eqref{CHpde1} with $\mu$ and the inner product of equation \eqref{CHpde2} with $\frac{\partial \phi}{\partial t}$, we have
\begin{equation}
\frac{d F}{d t} = \frac{d}{d t}\left\{\left(F_{b}(\phi), 1\right)+\frac{\lambda}{2}\|\nabla \phi\|^{2}\right\} = -(\mathbf{v} \cdot \nabla \phi, \mu) -M\|\nabla \mu\|^{2}.	
\end{equation}
Based on the divergence theorem and the boundary condition, the inner product of the interfacial term $(\mu \nabla \phi, \mathbf{v})$ in the NS system and the inner product of the convection term $(\mathbf{v} \cdot \nabla \phi, \mu)$ in the CH system is equal. These two terms characterize the energies transfer between the kinetic energy $E_k$ and the free energy $F$. We define $\text{transfer}_\mathrm{F-EK} := \int_{\Omega} (\mu \nabla \phi, \mathbf{v}) d \mathbf{x}$ as the energy quantity transferred from the free energy to the kinetic energy per unit time. The subscript $\mathrm{F-EK}$ means the energy transfer from $F$ to $E_k$ and vice versa. In the system, we have the relation
\begin{equation*}
\text{transfer}_\mathrm{F-EK} = - \text{transfer}_\mathrm{EK-F}.	
\end{equation*}
For the NS system, we have $\displaystyle (-\nabla p, \mathbf{v})= \int_{\partial \Omega} p \mathbf{v} \cdot \mathbf{n} d s+\int_{\Omega} p \nabla \cdot \mathbf{v} d \mathbf{x}=0$. The energy conservation laws for the CH system and the NS system can be concluded, respectively, as follows, 
$$\frac{d F}{d t}=\operatorname{transfer}_{\mathrm{EK}-\mathrm{F}}- M\|\nabla \mu\|^{2}, 
\quad \frac{d E_{k}}{d t} = \operatorname{transfer}_{\mathrm{F}-\mathrm{EK}} - \eta\|\nabla \mathbf{v}\|^{2}.$$ 	
Consequently, the desired energy law \eqref{NSCHPDEenergy} for the NSCH system is obtained.
\end{proof}

\section{SPH discretization and the NSCH ODE system}  \label{section:SPH&ODE}
\setcounter{equation}{0}
\renewcommand{\theequation}{\arabic{section}.\arabic{equation}}

In this section, some basics of the SPH method will be given first. Then a physical consistent NSCH ODE system will be derived based on the SPH discretization in space.

There are two main steps to the SPH method: the integral representation and the particle approximation. At the integral representation step, a physical field function $f(x)$ is converted into its integral representation based on the convolution of variables through a smoothed and weighted function, which gives
\begin{equation}
\langle f(\mathbf{x}) \rangle_h=\int_{\Omega} f\left(\mathbf{x}^{\prime}\right) W\left(\mathbf{x}-\mathbf{x}^{\prime}, h\right) d \mathbf{x}^{\prime}.
\end{equation}
Here, $\Omega$ is the supporting domain, and with a smoothing length $h$, a smoothing kernel function $W \left(\mathbf{x}-\mathbf{x}^{\prime},  h\right)$ is used to approximate the $\delta$ function in the convolution. The SPH approximation is represented by the symbol $\langle \rangle_h$, or by the subscript $h$, as in $f_{h}(\mathbf{x})$. The expression specifies the contribution to any field variable at position $\mathbf{x}$ by a particle at $\mathbf{x}^{\prime}$ that lies within the compact support of the kernel function. The kernel function becomes a Dirac delta function when $h \rightarrow 0$, namely $\lim _{h \rightarrow 0} W\left(\mathbf{x}-\mathbf{x}^{\prime}, h\right)=\delta\left(\mathbf{x}-\mathbf{x}^{\prime}\right)$, and it must satisfy the unity condition:
\begin{equation*}
\int_{\Omega} W\left(\mathbf{x}-\mathbf{x}^{\prime}, h\right) d \mathbf{x}^{\prime}=1.
\end{equation*}

Particle approximation: it is achieved by replacing the integral representation of the field function with the sum of all particle values. It also means that these material particles interact with each other, and the kernel function controls the range of their effects.  Function values at any position in the domain can be approximated through
\begin{equation}
f_{h}(\mathbf{x})= \sum_{j} f\left(\mathbf{x}_{j}\right) W\left(\mathbf{x}-\mathbf{x}_{j}, h\right) V_{j}
= \sum_{j} \frac{m_{j}}{\rho_{j}} f\left(\mathbf{x}_{j}\right) W\left(\mathbf{x}-\mathbf{x}_{j}, h\right),
\end{equation}
where $V_{j}$ is the volume element of one particle $j$ at $\mathbf{x}_{j}$, and it has been replaced by the ratio between the mass and density: $V_{j}=m_{j} / \rho_{j}$. If we consider the location of one particle $i$:
\begin{equation}
f_{h}(\mathbf{x}_i)=\sum_{j} \frac{m_{j}}{\rho_{j}} f\left(\mathbf{x}_{j}\right) W\left(\mathbf{x}_i-\mathbf{x}_{j}, h\right).
\end{equation}
Here, $i$ and $j$ are particle indexes, and particle $j$ is one of neighboring particles for $i$. The physical variable $f\left(\mathbf{x}_{i}\right)$ can be approximated using the interpolation formula with the corresponding physical variables of neighboring particles $f\left(\mathbf{x}_{j}\right)$. We also use $\mathbf{x}_{ij}$ and $r_{ij}$ to replace $\mathbf{x}_{i}-\mathbf{x}_{j}$ and $\|\mathbf{x}_{i}-\mathbf{x}_{j}\|$ respectively in the following content. It must be noted that the kernel function applied in this paper is the Quintic Spline SPH kernel.
\begin{equation*}
W\left(\mathbf{x}_{i j}, h\right)=\sigma_{d} \begin{cases}(3-q)^{5}-6(2-q)^{5}+15(1-q)^{5}, & 0 \leq q< 1,
 \\ (3-q)^{5}-6(2-q)^{5}, & 1 \leq q< 2,
 \\ (3-q)^{5}, & 2 \leq q< 3, \\ 0, & 3 \leq q.
\end{cases}
\end{equation*}
If we set the support radius of the domain of influence to be $h$, we have $\hat h = h/3$ and $q=\left\|\mathbf{x}_{i j}\right\| / \hat h$. $\sigma_{d}$ is the dimensional normalization factor, where $\sigma_{1}=1 / 120 \hat h, \sigma_{2}=7 / 478 \pi \hat h^{2}$ and $\sigma_{3}=3 / 359 \pi \hat h^{3}$ in 1D, 2D and 3D domains respectively.
\begin{table}[!htb]
\begin{center}
\caption{Some common SPH operators for discretizing PDE formulas}
\label{SPH operators table}
\renewcommand{\arraystretch}{1.5}
\begin{tabular}{l|l|l}
\hline  Divergence & $\nabla \cdot \mathbf{A}(\mathbf{x}_i)$ &
$\nabla_h \cdot \mathbf{A}(\mathbf{x}_i)=-\sum_{j} \frac{m_{j}}{\rho_{i}} \left(  \mathbf{A}_i - \mathbf{A}_j  \right) \cdot \nabla_i W_{ij}$  \\
\hline  Gradient & $\nabla f(\mathbf{x}_i)$        &
 $\nabla_{h} f\left(\mathbf{x}_{i}\right) = -\sum_{j} \frac{m_{j}}{\rho_{i}}\left(f_{i}-f_{j}\right) \nabla_{i} W_{ij}$  \\
\hline  Gradient (symmetric) &  $ \nabla \frac{f}{\rho}|_{\mathbf{x}=\mathbf{x}_i}$ &
$\nabla_h \frac{f}{\rho}|_{\mathbf{x}=\mathbf{x}_i}  = \sum_{j} m_{j}\left(\frac{f_{i}}{\rho_{i}^{2}}+\frac{f_{j}}{\rho_{j}^{2}}\right) \nabla_{i} W_{i j}$  \\
\hline Laplace (scalar/vector) & $\nabla^{2} f(\mathbf{x}_i)$  &
$\nabla_h^{2} f(\mathbf{x}_i) = 2 \sum_{j} \frac{m_{j}}{\rho_{j}} \frac{f(\mathbf{x}_{i})-f(\mathbf{x}_{j})}{\left\|\mathbf{x}_{i j}\right\|^{2}}\left(\mathbf{x}_{i j} \cdot \nabla_{i} W_{i j}\right)$  \\
\hline Laplace (vector) & $\nabla^{2} \mathbf{A}(\mathbf{x}_i)$  &
$\nabla_h^{2} \mathbf{A}(\mathbf{x}_i) = 2(d+2) \sum_{j} \frac{m_{j}}{\rho_{j}} \frac{\left(\mathbf{A}_{i}-\mathbf{A}_{j}\right) \cdot \mathbf{x}_{ij}}{\left\|\mathbf{x}_{ij}\right\|^{2}} \nabla_{i} W_{i j}$\\
\hline
\end{tabular}
\end{center}
\end{table}

It is essential to discretize the PDE equations through appropriate SPH operators to maintain the inheritance of physical properties like mass conservation, momentum conservation and energy dissipation law at the ODE and fully discrete level. Some common SPH operators are introduced in Table \ref{SPH operators table}. $W\left(\mathbf{x}_{ij}, h\right)$ is denoted by $W_{ij}$ as well and $\nabla_{i}W_{ij}$ is used to represent the gradient of the kernel function $W$ at the position of particle $i$. Therefore, the anti-symmetric property of gradient can be obtained:
\begin{equation}
\nabla_{j} W_{i j}=-\nabla_{i} W_{i j}.   \label{anti-symmetric}
\end{equation}
The number of dimensions is denoted by $d$ in the SPH Laplace operator. In the following, we will explain why the symmetric or anti-symmetric properties of the SPH operators or the SPH mathematical expression with physical consistency are vital. In order to avoid singularities, a minor term of $0.01h^2$ is added to the denominator of the SPH Laplace operator, where $\left\|\mathbf{x}_{ij}\right\|^{2}$ may vanish when $i$ is equal to $j$ as follows:
\begin{equation*}
\nabla_h^{2} f(\mathbf{x}_i) = 2 \sum_{j} \frac{m_{j}}{\rho_{j}} \frac{f(\mathbf{x}_{i})-f(\mathbf{x}_{j})}{\left\|\mathbf{x}_{i j}\right\|^{2} + 0.01h^2}\left(\mathbf{x}_{i j} \cdot \nabla_{i} W_{i j}\right).
\end{equation*}
Based on the anti-symmetric property of gradient as \eqref{anti-symmetric}, the important negative property of the kernel function can be deduced:
\begin{equation}
\nabla_{i} W_{i j}=\mathbf{x}_{i j} \omega\left( r_{ij}\right), \label{negative}
\end{equation}
where $\omega$ is a scalar function and $\omega \leq 0$ always. 
This negative property can be easily seen from the Figure \ref{kernel function} as well.
\begin{figure}[h]
\centering  
\subfigure[]{
\includegraphics[width=0.35\textwidth]{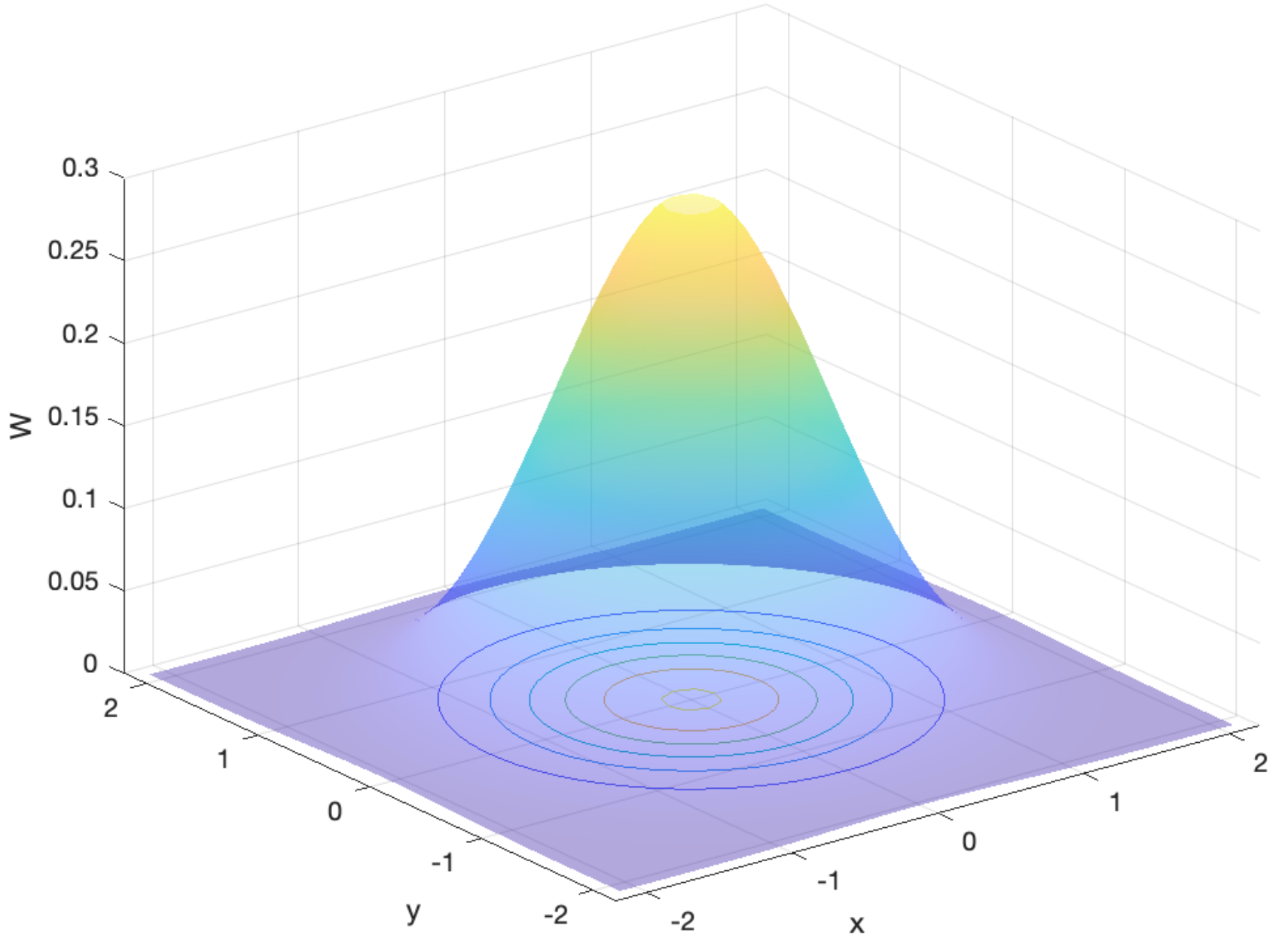}}
\subfigure[]{
\includegraphics[width=0.35\textwidth]{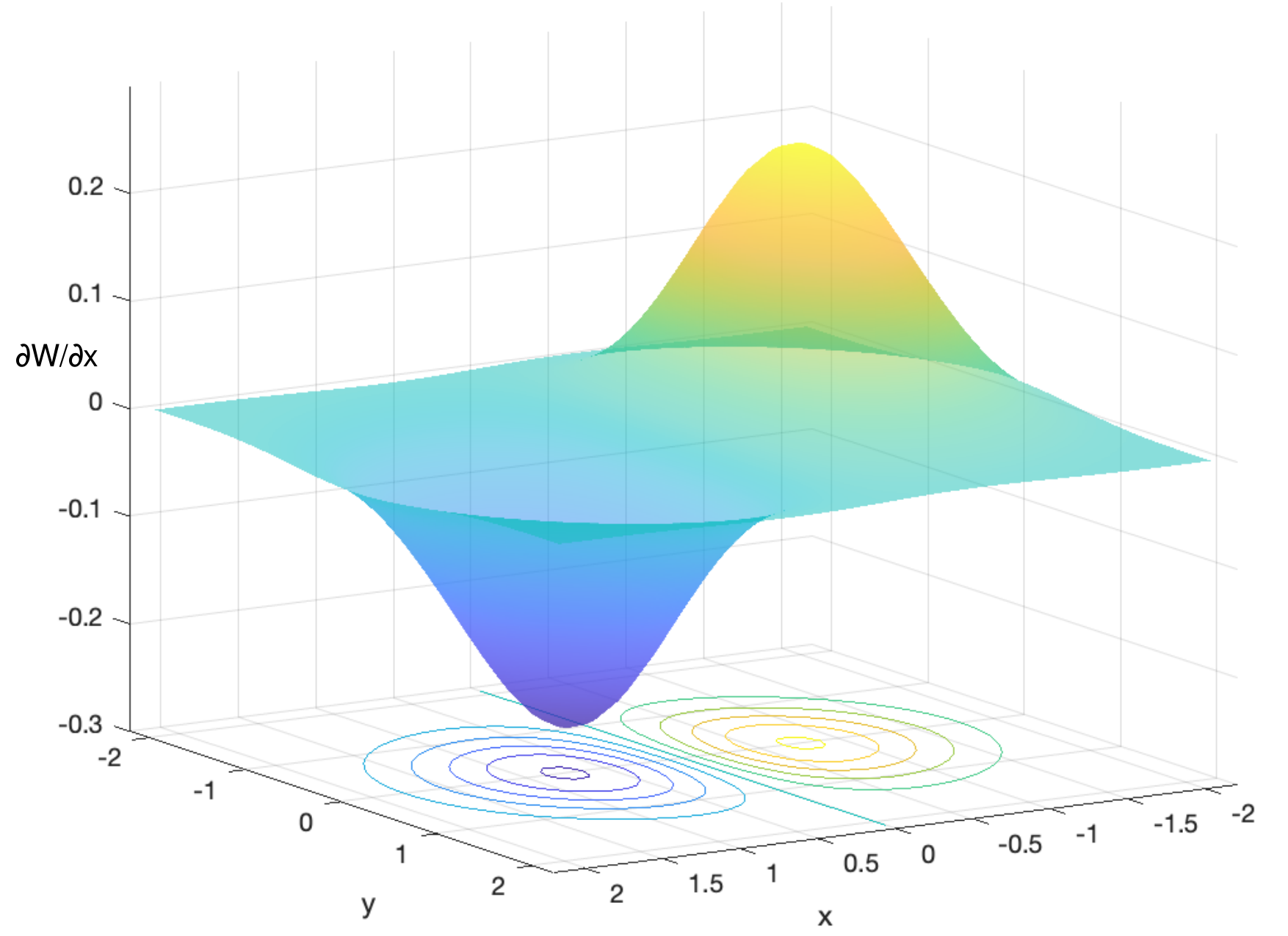}}
\subfigure[]{
\includegraphics[width=0.35\textwidth]{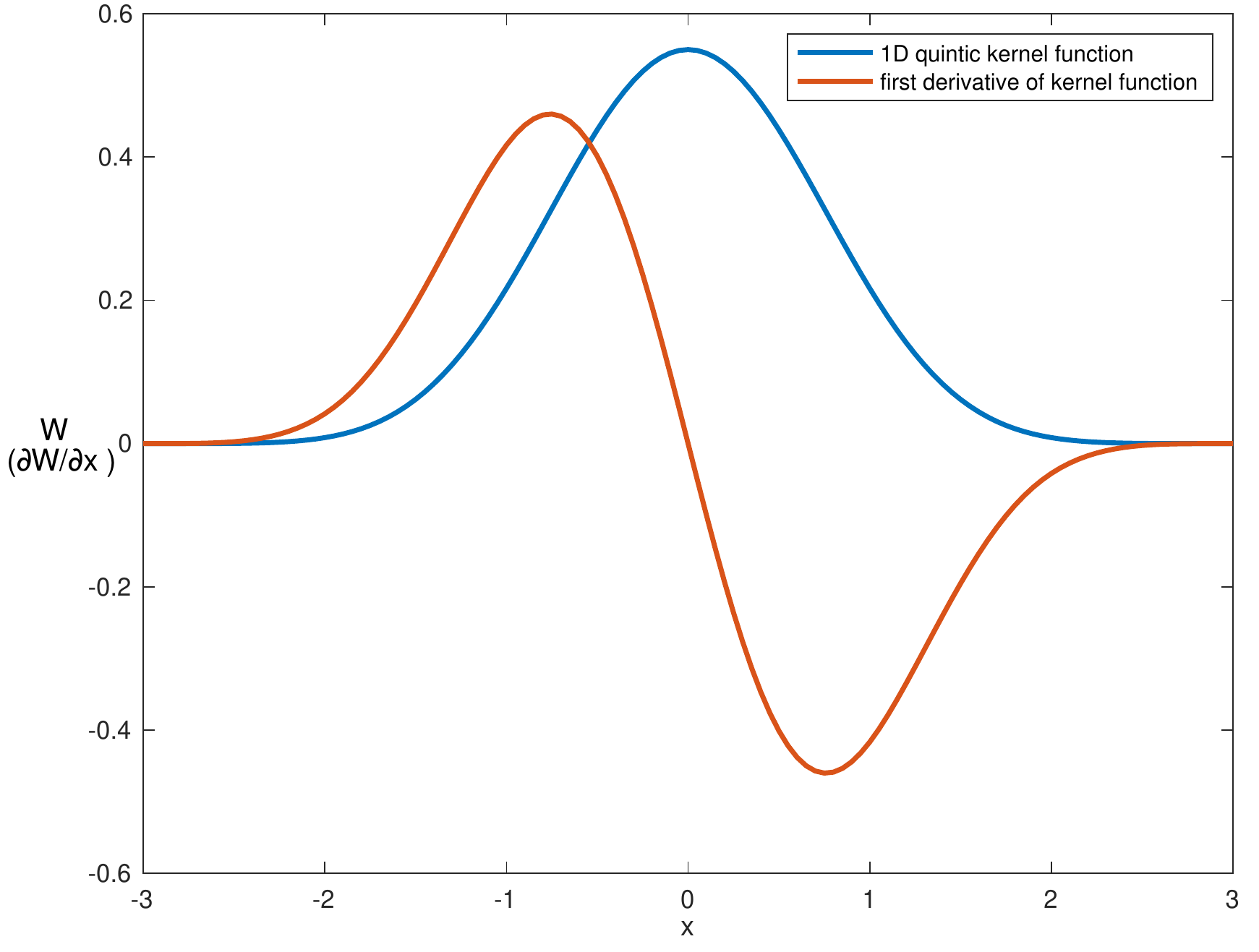}}
\caption{(a) Contour graph for Quintic kernel in 2D; (b) First derivative of kernel function; (c) Plot of Quintic kernel function and its first derivative in 1D}
\label{kernel function}
\end{figure}

From \eqref{anti-symmetric} and \eqref{negative}, we can derive the following lemma. We omit the proof and leave it to the interested readers.
\begin{lem}
Assume the SPH operator $\sigma_{ij}$ is symmetric $(
\sigma_{ij} = \sigma_{ji})$, we obtain
\begin{equation}
\sum_{i} \sum_{j} m_{i} m_{j} \sigma_{i j} \nabla_{i} W_{i j} = 0,  \label{SPHsymmetric}
\end{equation}
and if $\sigma_{i j} \geq 0$, we have
\begin{equation} \label{SPHnegative}
\left\{\begin{array}{l}
\sum_{i} \sum_{j} m_{i} m_{j}  \sigma_{i j}\left(\mathbf{x}_{i j} \cdot \nabla_{i} W_{i j}\right) \leq 0,
 \\
\sum_{i} \sum_{j} m_{i} m_{j} \sigma_{i j}\left(\mathbf{v}_{i j} \cdot \mathbf{x}_{i j}\right) \mathbf{v}_{i} \cdot \nabla_{i} W_{i j} \leq 0.
\end{array}\right.	
\end{equation}
\end{lem}

\subsection{The NSCH ODE system} \label{subsection:ODE}
Firstly, we review the NSCH PDE system in Lagrangian form. Using the identity $\nabla \cdot(\nabla \phi \otimes \nabla \phi)=\Delta \phi \nabla \phi+\frac{1}{2} \nabla\|\nabla \phi\|^{2}$, the PDE system can be rewritten as
\begin{subequations}
\begin{eqnarray}
\frac{D \phi}{D t}=\nabla \cdot(M \nabla \mu), \\
\mu= \lambda_{b}\left(\phi^{3}-\phi\right) -\lambda \Delta \phi, \\
\frac{D \mathbf{v}}{D t}=-\nabla\frac{p}{\rho}+ \frac{\eta}{\rho} \Delta \mathbf{v}+\frac{\mu_{\nabla} \nabla \phi}{\rho}, \\
\nabla \cdot \mathbf{v}=0.
\end{eqnarray} \label{rewritePDE}
\end{subequations}
The design of ODE expression starts from the definition of free energy at the ODE level. The homogenous and non-homogenous parts of total free energy are
\begin{equation*} 
\left\{\begin{array}{l}
F_{h}^{b}=F_{h}^{b}\left(\phi_{1}, \phi_{2}, \ldots, \phi_{n}\right), \\
F_{h}^{\nabla}=F_{h}^{\nabla}\left(\phi_{1}, \phi_{2}, \ldots, \phi_{n}, \mathbf{x}_{1}, \mathbf{x}_{2}, \ldots, \mathbf{x}_{n}\right), \\
F_{h} = F_{h}^{b} + F_{h}^{\nabla}.
\end{array}\right.	
\end{equation*}
To keep the local and non-local properties respectively, the energy at the ODE level can be defined through the summation of particles one by one as follows:
\begin{equation}
\left\{\begin{array}{l}
F_{h}^{b}:=\sum_{i} \frac{m_{i}}{\rho_{i}} f_{b}\left(\phi_{i}\right), \\
F_{h}^{\nabla}:=\sum_{i} \frac{m_{i}}{\rho_{i}} \frac{\lambda}{2}\left(\nabla_{h} \phi_{i}\right) \cdot\left(\nabla_{h} \phi_{i}\right).
\end{array}\right. \label{F_definition}
\end{equation}
%
%
Then, the free energy change with time in the ODE form is given as
 \begin{equation}
\frac{d F_h}{d t} = \frac{d F_{h}^{b} }{d t} + \frac{d F_{h}^{\nabla}}{d t}  =
\sum_{i} \frac{m_{i}}{\rho_{i}} \frac{\partial f_{b}\left(\phi_{i}\right)}{\partial \phi_{i}} \frac{d \phi_{i}}{d t}  +
\sum_{i} \frac{\partial F_h^{\nabla}}{\partial \phi_{i}} \frac{d \phi_{i}}{d t}+
\sum_{i} \frac{\partial F_{h}^{\nabla}}{\partial \mathbf{x}_{i}} \frac{\partial \mathbf{x}_{i}}{\partial t}. \label{F3part}
\end{equation}
The term $\sum_{i} \frac{\partial F_{h}^{\nabla}}{\partial \mathbf{x}_{i}} \frac{\partial \mathbf{x}_{i}}{\partial t}$ is responsible for energy transferred from kinetic energy $E_k$ to free energy $F$ at the ODE level, which gives
\begin{equation}
\langle \text{transfer}_\mathrm{EK-F} \rangle_h  =  \sum_{i} \frac{\partial F_{h}^{\nabla}}{\partial \mathbf{x}_{i}} \frac{\partial \mathbf{x}_{i}}{\partial t}. \label{ODEek2f}
\end{equation}
The energy-transfer term is originated from the inhomogeneity (the interface region) and other two terms control the energy dissipation coming from the diffusive effect of the CH equation. Since the chemical potential $\mu_{i}^{b}$ locally depends on $\phi_i$ as $\mu_{i}^{b}:=\frac{\partial f_{b}\left(\phi_{i}\right)}{\partial \phi_{i}}$ and $\mu_{i}^{\nabla}$ is a non-local function given by $\mu_{i}^{\nabla}:=\frac{\rho_{i}}{m_{i}} \frac{\partial F_{h}^{\nabla}}{\partial \phi_{i}}$, $\mu$ is defined at the ODE level as follows:
\begin{equation}
\mu_{i}= \frac{\partial f_{b}\left(\phi_{i}\right)}{\partial \phi_{i}} + \frac{\rho_{i}}{m_{i}} \frac{\partial F_{h}^{\nabla}}{\partial \phi_{i}}.
\label{ODEmudef}
\end{equation}
%
%
When we turn to the momentum equation of the NS system, in order to keep the physical consistency of energy transfer between the CH system and the NS system, the interfacial force should be defined as
\begin{equation}
\boldsymbol{\mathcal { T }}_{i}:= \frac{\partial F_{h}^{\nabla}}{\partial \mathbf{x}_{i}}. \label{interfacial force}
\end{equation}
This NSCH ODE system guarantees these physical properties and it is finalized as
\begin{subequations}
 \begin{gather}
\begin{aligned}
\frac{d \phi_{i}}{d t}= 2M \sum_{j} \frac{m_{j}}{\rho_{j}} \frac{\mu_{i}-\mu_{j}}{\left\|\mathbf{x}_{i j}\right\|^{2}}\left(\mathbf{x}_{i j} \cdot \nabla_{i} W_{i j}\right), \label{ODECH1}
\end{aligned}
\\
\begin{aligned}
\mu_{i}=\frac{\rho_{i}}{m_{i}} \frac{\partial F_{h}^{\nabla}}{\partial \phi_{i}}+f_{b}^{\prime}\left(\phi_{i}\right),\label{ODECH2}
\end{aligned}
\\
\begin{aligned}
\rho_{i}=\sum_{j} m_{j} W_{i j}, \label{ODERHO}
\end{aligned}
\\
\begin{aligned}
\frac{d \mathbf{v}_{i}}{d t} =& -\frac{\boldsymbol{\mathcal { T }}_{i}}{m_{i}}-\sum_{j} m_{j}\left(\frac{p_{i}}{\rho_{i}^{2}}+\frac{p_{j}}{\rho_{j}^{2}}\right) \nabla_{i} W_{i j}
+ (d+2) \sum_{j} \frac{m_{j}}{\rho_{i} \rho_{j}} \frac{\left(\eta_{i}+\eta_{j}\right) \left( \mathbf{v}_{i j} \cdot \mathbf{x}_{i j} \right) }{\left\|\mathbf{x}_{i j}\right\|^{2}} \nabla_{i} W_{i j}, \label{ODENS}
\end{aligned}
\\
\begin{aligned}
-\frac{1}{\rho_{i}} \sum_{j} m_{j}\left(\mathbf{v}_{i}-\mathbf{v}_{j}\right) \cdot \nabla_{i} W_{i j}=0. \label{ODEdiv}
\end{aligned}
\end{gather}
\end{subequations}
Here and after, the notation $\mathbf{v}_{i j}$ is equivalent to $\mathbf{v}_{i}-\mathbf{v}_{j}$. The appearance of the negative sign of the interfacial term $-\frac{\boldsymbol{\mathcal { T }}_{i}}{m_{i}}$ in the ODE system is consistent with the variational derivative of PDE in $H^{-1}$ space. The viscosity is chosen an average to ensure the symmetric form.

\subsection{Energy dissipation law at the ODE level}
The energy dissipation property of the given ODE system will be proved. Similar to the verification of the PDE model’s energy dissipation law, the inner product is utilized here.
\begin{thm} \label{ODEenergyThm}
Assume that the NSCH ODE system based on SPH discretization as \eqref{ODECH1}, \eqref{ODECH2}, \eqref{ODERHO}, \eqref{ODENS} and \eqref{ODEdiv} is employed, the energy dissipation law at the ODE level can be maintained as follows:
\begin{equation}
\begin{aligned} \left\langle\frac{d E_{\text {total }}}{d t}\right\rangle_{h} &=\left\langle\frac{d E_{k}}{d t}\right\rangle_{h}+\left\langle\frac{d F}{d t}\right\rangle_{h} \leq 0. \label{ODEenergyTotal}
\end{aligned}	
\end{equation}
\end{thm}
\begin{proof}
\textbf{\textit{Conservation of energy transfer:}} We have defined the interfacial force as \eqref{interfacial force} and the chemical potential that comes from inhomogeneity as \eqref{ODEmudef}. Then by taking the inner product of the time derivative term with $m_i \mathbf{v}_{i}$, we get the change of kinetic energy $E_k$ per time in the system:
\begin{equation}
\left\langle\frac{d E_{k}}{d t}\right\rangle_{h}=\sum_{i} \left( m_{i} \mathbf{v}_{i} , \frac{d \mathbf{v}_{i}}{d t} \right)=\sum_{i} \frac{d\left(\frac{1}{2} m_{i}\left\|\mathbf{v}_{i}\right\|^{2}\right)}{d t}.	 \label{ODE_Ekterm}
\end{equation}
The interfacial term $-\frac{\boldsymbol{\mathcal { T }}_{i}}{m_{i}}$ in the momentum equation governs the energy transferred from free energy $F$ and the kinetic energy $E_k$. By taking the inner product of this term with $m_i \mathbf{v}_{i}$ and adding up the total for all particles, we obtain
\begin{equation}
\langle \text{transfer}_\mathrm{F-EK} \rangle_h = \sum_{i} \left( -\frac{\boldsymbol{\mathcal { T }}_{i}}{m_{i}} , m_{i} \mathbf{v}_{i} \right)
= - \sum_{i} \frac{\partial F_{h}^{\nabla}}{\partial \mathbf{x}_{i}} \frac{\partial \mathbf{x}_{i}}{\partial t}. \label{ODE_INTERterm}
\end{equation}
Thus, we can conclude that
\begin{equation}
\langle \text{transfer}_\mathrm{EK-F} \rangle_h
= - \langle \text{transfer}_\mathrm{F-EK} \rangle_h
=\sum_{i} \frac{\partial F_{h}^{\nabla}}{\partial \mathbf{x}_{i}} \frac{\partial \mathbf{x}_{i}}{\partial t}. \label{ODEenergytransfer}
\end{equation}
The conservation of energy transfer between kinetic energy $E_k$ and free energy $F$ is examined.

\textbf{\textit{Conservation of energy in the CH system:}}
If we take the inner product of \eqref{ODECH1} with $\mu_{i}$ and integrate it throughout the whole domain, taking into account the particle's control volume $\frac{m_{i}}{\rho_{i}}$, the results turn out easily by swapping the index $i$ and $j$ and taking advantage of the SPH negative property as \eqref{negative}.
\begin{equation}
\begin{split}
\sum_{i} \frac{m_{i}}{\rho_{i}}\left( \frac{d \phi_{i}}{d t}, \mu_{i}\right)
&= \sum_{i} \frac{m_{i}}{\rho_{i}}\left( 2M \sum_{j} \frac{m_{j}}{\rho_{j}} \frac{\mu_{i}-\mu_{j}}{\left\|\mathbf{x}_{i j}\right\|^{2}}  \left(\mathbf{x}_{i j} \cdot \nabla_{i} W_{i j}\right) , \mu_{i}\right) \\
&=M \sum_{i} \sum_{j} \frac{m_{i} m_{j}}{\rho_{i} \rho_{j}} \frac{\left(\mu_{i}-\mu_{j}\right)^{2}}{\left\|\mathbf{x}_{i j}\right\|^{2}}   \left(\mathbf{x}_{i j} \cdot \nabla_{i} W_{i j}\right) \\
&=M \sum_{i} \sum_{j} \frac{m_{i} m_{j}}{\rho_{i} \rho_{j}}\left(\mu_{i}-\mu_{j}\right)^{2} \omega \left(r_{ij}\right) \leq 0. \label{ODEequalCH1}
\end{split}
\end{equation}
Then, by taking the inner product of \eqref{ODECH2} with $\frac{d \phi_{i}}{d t}$ and summing for all particles and based on the equation \eqref{F3part}, we find
\begin{equation}
\begin{split}
\sum_{i} \frac{m_{i}}{\rho_{i}}\left(\mu_{i}, \frac{d \phi_{i}}{d t}\right)
&=\sum_{i} \frac{m_{i}}{\rho_{i}}\left(\frac{\rho_{i}}{m_{i}} \frac{\partial F_{h}^{\nabla}}{\partial \phi_{i}}, \frac{d \phi_{i}}{d t}\right)+\sum_{i} \frac{m_{i}}{\rho_{i}}\left(\frac{\partial f_{b}\left(\phi_{i}\right)}{\partial \phi_{i}}, \frac{d \phi_{i}}{d t}\right) \\
&=\left\langle\frac{d F_{b}}{d t}\right\rangle_{h}+\left\langle\frac{d F_{\nabla}}{d t}\right\rangle_{h} = \left\langle\frac{d F}{d t}\right\rangle_{h}-\langle\text{transfer}_\mathrm{F-EK}\rangle_{h}.
\label{ODEequalCH2}
\end{split}
\end{equation}
Combining equations \eqref{ODEequalCH1} and \eqref{ODEequalCH2}, it reaches the conclusion:
\begin{equation}
\left\langle\frac{d F}{d t}\right\rangle_{h}=\langle\text {transfer}_\mathrm{EK-F}\rangle_{h}+M \sum_{i} \sum_{i} \frac{m_{i} m_{j}}{\rho_{i} \rho_{j}}\left(\mu_{i}-\mu_{j}\right)^{2} \omega\left(r_{ij}\right). \label{ODECHenergy}
\end{equation}

\textbf{\textit{Conservation of energy in the NS system:}}
We take the inner product of the momentum equation \eqref{ODENS} with $m_{i} \mathbf{v}_{i}$ and do summation for all particles. The change of kinetic energy $E_k$ at the ODE level is given in \eqref{ODE_Ekterm}. The divergence free condition \eqref{ODEdiv} yields
\begin{equation*}
-\sum_{i} m_{i} \frac{p_{i}}{\rho_{i}^{2}} \sum_{j} m_{j}\left(\mathbf{v}_{i}-\mathbf{v}_{j}\right) \cdot \nabla_{i} W_{i j}
= -\sum_{i} \sum_{j} m_{i} m_{j} \frac{p_{i}}{\rho_{i}^{2}}\left(\mathbf{v}_{i}-\mathbf{v}_{j}\right) \cdot \nabla_{i} W_{i j} = 0.
\end{equation*}
Let $I_1, I_2$ denote the pressure term $\sum_{i} \left(m_{i} \mathbf{v}_{i}, -\sum_{j} m_{j}\left(\frac{p_{i}}{\rho_{i}^{2}}+\frac{p_{j}}{\rho_{j}^{2}}\right) \nabla_{i} W_{i j} \right)$ and the viscosity term $\sum_{i} \left( m_{i} \mathbf{v}_{i} , \frac{\eta_i+\eta_j}{2\rho_i} \nabla^2_h \mathbf{v}_{i} \right)$, respectively. Then we have
\begin{equation}
I_1= -\frac{1}{2} \sum_{i} \sum_{j} m_{i} m_{j}\left(\frac{p_{i}}{\rho_{i}^{2}}+\frac{p_{j}}{\rho_{j}^{2}}\right)\left(\mathbf{v}_{i}-\mathbf{v}_{j}\right) \cdot \nabla_{i} W_{i j} = 0, \label{ODE_pressureterm}
\end{equation}
%
and
\begin{equation}
\begin{split}
 I_2&= \frac{(d+2)}{2} \sum_{i} \sum_{j} \frac{m_{i} m_{j}}{\rho_{i} \rho_{j}} \frac{\left(\eta_{i}+\eta_{j}\right) \mathbf{v}_{i j} \cdot \mathbf{x}_{i j}}{\left\|\mathbf{x}_{i j}\right\|^{2}}\left(\mathbf{v}_{i}-\mathbf{v}_{j}\right) \cdot \nabla_{i} W_{i j} \\
 &= \frac{(d+2)}{2} \sum_{j} \frac{m_{i}m_{j}}{\rho_{i} \rho_{j}} \frac{\left(\eta_{i}+\eta_{j}\right)}{\left\|\mathbf{x}_{i j}\right\|^{2}}\|\mathbf{v}_{i j} \cdot \mathbf{x}_{i j}\|^{2} \omega\left(r_{ij}\right) \leq 0.  \label{ODE_visterm}
\end{split}
\end{equation}
%
Combining \eqref{ODE_Ekterm}, \eqref{ODE_INTERterm}, \eqref{ODE_pressureterm} and \eqref{ODE_visterm}, we obtain
\begin{equation}
\left\langle\frac{d E_{k}}{d t}\right\rangle_{h}=\langle\text {transfer}_\mathrm{F-EK}\rangle_{h}+ \frac{(d+2)}{2} \sum_{j} \frac{m_{i}m_{j}}{\rho_{i} \rho_{j}} \frac{\left(\eta_{i}+\eta_{j}\right)}{\left\|\mathbf{x}_{i j}\right\|^{2}}\|\mathbf{v}_{i j} \cdot \mathbf{x}_{i j}\|^{2} \omega\left(r_{ij}\right). \label{ODENSenergy}
\end{equation}
%
Finally, combining \eqref{ODEenergytransfer}, \eqref{ODECHenergy} and \eqref{ODENSenergy}, we complete the proof.
\end{proof}

\subsection{Momentum conservation of the ODE system}
Next, we turn to the proof of the momentum conservation holds as well for the designed ODE system.
\begin{thm} \label{theorem_ODEmomentum}
The NSCH ODE system which consists of equations \eqref{ODECH1}, \eqref{ODECH2}, \eqref{ODERHO}, \eqref{ODENS} and \eqref{ODEdiv} can maintain momentum conservation, namely
\begin{equation}
\sum_{i} m_{i} \frac{d \mathbf{v}_{i}}{d t}=0. \label{ODEmomentum}
\end{equation}
\end{thm}
\begin{proof}
In the first place, by taking the inner product of momentum equation \eqref{ODENS} with $m_i$ and doing summation for all particles, we have
\begin{equation}
\begin{split}
\sum_{i} m_{i} \frac{d \mathbf{v}_{i}}{d t}=&-\sum_{i} \sum_{j} m_{i} m_{j}\left(\frac{p_{i}}{\rho_{i}^{2}}+\frac{p_{j}}{\rho_{i}^{2}}\right) \nabla_{i} W_{i j}\\
&+(d+2) \sum_{i} \sum_{j} \frac{m_{i} m_{j}}{\rho_{i} \rho_{j}} \frac{\left(\eta_{i}+\eta_{j}\right) \mathbf{v}_{i j} \cdot \mathbf{x}_{i j}}{\left\|\mathbf{x}_{i j}\right\|^{2}} \nabla_{i} W_{i j}
-\sum_{i} \boldsymbol{\mathcal { T }}_{i}.
\end{split}
\end{equation}
Two above SPH operators can be treated as a whole symmetric one ($\sigma_{ij}=\sigma_{ji}$), like $\sigma_{i j}:=-\left(\frac{p_{i}}{\rho_{i}^{2}}+\frac{p_{j}}{\rho_{j}^{2}}\right)+\frac{(d+2)}{\rho_{i} \rho_{j}} \frac{\left(\eta_{i}+\eta_{j}\right) \mathbf{v}_{i j} \cdot \mathbf{x}_{i j}}{\left\|x_{i j}\right\|^{2}}$. Then, based on the property \eqref{SPHsymmetric}, we find
\begin{equation}
\sum_{i} \sum_{j} m_{i} m_{j} \sigma_{i j} \nabla_{i} W_{i j}=0.
\end{equation}
By proving $\sum_{i} \mathcal{T}_{i}=\sum_{i} \frac{\partial F_{h}^{\nabla}}{\partial \mathbf{x}_{i}}=0$, we can conclude momentum conservation. It must be noted that the selected SPH gradient operator for $\phi$ is
\begin{equation*}
\nabla_{h} \phi_i = -\sum_{j} \frac{m_{j}}{\rho_{i}}\left(\phi_{i}-\phi_{j}\right) \nabla_{i} W_{i j}.
\end{equation*}
Based on the definition of $F_{h}^{\nabla}$ as \eqref{F_definition}, it becomes
\begin{equation}
\begin{split}
F^{\nabla}_h
&=\sum_{k} \frac{m_{k}}{\rho_{k}} \frac{\lambda}{2}\left(\nabla_{h} \phi_{k}\right) \cdot\left(\nabla_{h} \phi_{k}\right)\\
&=\sum_{k} \frac{m_{k}}{\rho_{k}} \frac{\lambda}{2}\left[-\sum_{l} \frac{m_{l}}{\rho_{l}}\left(\phi_{k}-\phi_{l}\right) \nabla_{k} W_{k l}\right] \cdot\left[-\sum_{q} \frac{m_{q}}{\rho_{q}}\left(\phi_{k}-\phi_{q}\right) \nabla_{k} W_{k q}\right].	\label{F_grad_sph}
\end{split}
\end{equation}
We consider the compressible flow here. Hence, we can set: $m_{i} \equiv m$  and  $\rho_{i} \equiv \rho$ and define $f_{k l q} \equiv\left(\phi_{k}-\phi_{l}\right)\left(\phi_{k}-\phi_{q}\right) \nabla_{k} W_{k l} \cdot \nabla_{k} W_{k q}$, then
\begin{equation}
F_{h}^{\nabla} = \frac{\lambda m^{3}}{2 \rho^{3}} \sum_{k} \sum_{l} \sum_{q} f_{k l q}.	
\end{equation}
Correspondingly, we get
\begin{equation}
\sum_{i} \boldsymbol{\mathcal { T }}_{i}  = \sum_{i} \frac{\partial F_{h}^{\nabla}}{\partial \mathbf{x}_{i}}=\frac{\lambda m^{3}}{2 \rho^{3}} \sum_{k} \sum_{l} \sum_{q} \sum_{i} \frac{\partial f_{k l q}}{\partial \mathbf{x}_{i}}.
\end{equation}
Using the negative property of the kernel function, $f_{klq}$ can be expanded:
\begin{equation*}
\begin{aligned}
\frac{\partial f_{k l q}}{\partial \mathbf{x}_{i}}
& = \left(\phi_{k}-\phi_{l}\right)\left(\phi_{k}-\phi_{q}\right)\left\{\frac{\partial\left(\mathbf{x}_{k l} \omega\left(r_{k l}\right)\right)}{\partial \mathbf{x}_{i}} \mathbf{x}_{k q} \omega\left(r_{k q}\right)+\frac{\partial\left(\mathbf{x}_{k q} \omega\left(r_{k q}\right)\right)}{\partial \mathbf{x}_{i}} \mathbf{x}_{k l} \omega\left(r_{k l}\right)\right\} \\
&=\left( \phi_{k} -\phi_{l}\right)\left(\phi_{k}-\phi_{q}\right)\left(\delta_{k i}-\delta_{l i}\right)\left\{\omega\left(r_{k l}\right)+r_{k l} \omega^{\prime}\left(r_{k l}\right)\right\} \omega\left(r_{k q}\right) \mathbf{x}_{k q} \\
& \quad ... \quad +\left(\phi_{k}-\phi_{l}\right)\left(\phi_{k}-\phi_{q}\right)\left(\delta_{k i}-\delta_{q i}\right)\left\{\omega\left(r_{k q}\right)+r_{k q} \omega^{\prime}\left(r_{k q}\right)\right\} \omega\left(r_{k l}\right) \mathbf{x}_{k l}.
\end{aligned}	 
\end{equation*}
With the following equations of summation for the $i$ index:
\begin{equation*}
\left\{\begin{array}{l}
\sum_{i} \delta_{ki}=1 ,\\
\sum_{i}\left(\delta_{ki}-\delta_{li}\right)=0   \quad \text{and}  \quad  \sum_{i}\left(\delta_{k i}-\delta_{q i}\right)=0.
\end{array}\right.
\end{equation*}
We reach $\sum_{i} \frac{\partial f_{k l q}}{\partial x_{i, \alpha}} =0.	$
Finally, we get $\sum_{i} \boldsymbol{\mathcal { T }}_{i} = 0$ and \eqref{ODEmomentum} holds naturally.
\end{proof}

\begin{rem}
Through appropriate selections of SPH operators and the derivation based on the original energy and physical consistency, the proposed NSCH ODE system inherits the physical consistency, including mass conservation, momentum conservation, and the energy dissipation law. The spatial discretization has been completed based on the SPH framework. 
\end{rem}

\section{Fully discrete scheme} \label{section:Scheme}

Based on the energy-stable ODE system we proposed in the last section, we now focus on the temporal scheme. In this section, we decouple the CH system and the NS system and solve them separately. The projection method is adopted for the solution of the NS system. Some techniques like the convex-concave splitting approach, stabilizing term, and energy-factorization method are also utilized to construct the energy-stable fully discrete scheme. 
\subsection{Temporal discretization}
Now the fully discrete scheme of NSCH system reads as follows:
\begin{subequations}
 \begin{gather}
\begin{aligned}
\frac{\phi_{i}^{k+1}-\phi_{i}^{k}}{\Delta t}=2 M \sum_{j} \frac{m_{j}}{\rho_{j}} \frac{\mu_{i}^{k+1}-\mu_{j}^{k+1}}{\left\|\mathbf{x}_{i j}\right\|^{2}}\left(\mathbf{x}_{i j} \cdot \nabla_{i} W_{i j}\right), \label{FullCH1}
\end{aligned}
\\
\begin{aligned}
\mu_{i}^{k+1}=
\mu_{i}^{b}\left(\phi_{i}^{k}, \phi_{i}^{k+1}\right) + \mu_{i}^{\nabla}\left(\phi^{k+1}\right)
 =f_{b}^{\prime}\left(\phi_{i}^{k}, \phi_{i}^{k+1}\right) \label{FullCH2}
\end{aligned} + \frac{\rho_{i}}{m_{i}} \frac{\partial F_{h}^{\nabla}\left(\phi^{k+1}\right)}{\partial \phi_{i}^{k+1}},
\\
\begin{aligned}
\rho_{i}^{k+1}=\sum_{j} m_{j} W_{i j}, \label{Fullrho}
\end{aligned}
\\
\begin{aligned}
\frac{\tilde{\mathbf{v}}_{i}^{k+1}-\mathbf{v}_{i}^{k}}{\Delta t}=-\frac{\left\langle\boldsymbol{\mathcal { T }}_{i}\left(\tilde{\mathbf{x}}^{k+1}, \mathbf{x}^{k}\right)\right\rangle_{h}}{m_{i}}
+ (d+2) \sum_{j} \frac{m_{j}}{\rho_{i}^{k+1} \rho_{j}^{k+1}} \frac{\left(\eta_{i}+\eta_{j}\right)\left(\tilde{\mathbf{v}}_{ij}^{k+1} \cdot \mathbf{x}_{ij} \right) }{\left\|\mathbf{x}_{i j}\right\|^{2}} \nabla_{i} W_{i j}, \label{Fullmomentum}
\end{aligned}
\\
\begin{aligned}
\frac{\mathbf{v}_{i}^{k+1}-\tilde{\mathbf{v}}_{i}^{k+1}}{\Delta t}=-\sum_{j} m_{j}\left(\frac{p_{i}^{k+1}}{\left(\rho_{i}^{k+1}\right)^{2}}+\frac{p_{j}^{k+1}}{\left(\rho_{j}^{k+1}\right)^{2}}\right) \nabla_{i} W_{i j}, \label{FullPo}
\end{aligned}
\\
\begin{aligned}
\nabla_{h} \cdot \mathbf{v}_{i}^{k+1}=-\frac{1}{\rho_{i}^{k+1}} \sum_{j} m_{j}\left(\mathbf{v}_{i}^{k+1}-\mathbf{v}_{j}^{k+1}\right) \cdot \nabla_{i} W_{i j}=0, \label{Fulldiv}
\end{aligned}
\\
\begin{aligned}
\mathbf{x}_{i}^{k+1}-\mathbf{x}_{i}^{k}=\mathbf{v}_{i}^{k+1} \Delta t,
\label{FullUpdate}
\end{aligned}
\end{gather}
\end{subequations}
where $\phi$ represents the collection of phase function $\phi_i$ on each particle and $\mathbf{x}$ represents the collection of position information $\mathbf{x}_i$ on each particle.

Some terms in the fully discrete system require special treatment to ensure energy stability. The details of them are presented here. A semi-implicit scheme applied to term $\mu_i^{b}\left(\phi_i^{k}, \phi_i^{k+1}\right)$ reads:
\begin{equation}
\mu_i^{b}\left(\phi_i^{k}, \phi_i^{k+1}\right) = f_{b}^{\prime}\left(\phi_{i}^{k}, \phi_{i}^{k+1}\right)=\lambda_{b}\left[\left(\xi+\left(\phi_{i}^{k}\right)^{2}\right) \phi_{i}^{k+1}-(\xi+1) \phi_{i}^{k}\right]. \label{Fullmub}
\end{equation}
where $\xi$ is the energy-factorization parameter in the EF method. It results in a linear scheme for determining $\phi$. Based on the expression of $F^{\nabla}_h$ as \eqref{F_grad_sph}, a fully implicit scheme applied to term $\mu_i^{\nabla}(\phi^{k+1})$ reads:
\begin{equation}
\mu_{i}^{\nabla}\left(\phi^{k+1}\right) = \frac{\rho_{i}}{m_{i}}  \sum_{j} \frac{m_{j}}{\rho_{j}} \lambda \left[\sum_{p} \frac{m_{p}}{p_{j}}\left(\delta_{i j}-\delta_{i p}\right) \nabla_{j} W_{j p}\right]\left[\sum_{l} \frac{m_{l}}{\rho_{j}}\left(\phi_{j}^{k+1}-\phi_{l}^{k+1}\right) \nabla_{j}W_{jl} \right]. \label{Fullmunabla}
\end{equation}

Then, we utilize the stabilizing method for the interfacial term $\left\langle\boldsymbol{\mathcal { T }}_{i}\left(\tilde{\mathbf{x}}^{k+1}, \mathbf{x}^{k}\right)\right\rangle_{h}$. Based on the convex-concave splitting approach, we have
\begin{equation}
\left\langle\boldsymbol{\mathcal{T}}_{i}\left(\tilde{\mathbf{x}}^{k+1}, \mathbf{x}^{k}\right)\right\rangle_{h} = \frac{2\mathcal{M}}{\Delta t} \tilde{\mathbf{x}}_{i}^{k+1}+\boldsymbol{\mathcal {T}}_{i}\left(\mathbf{x}^{k}\right)- \frac{2\mathcal{M}}{\Delta t} \mathbf{x}_{i}^{k} = 2\mathcal{M} \tilde{\mathbf{v}}_{i}^{k+1} +\boldsymbol{\mathcal {T}}_{i}\left(\mathbf{x}^{k}\right), \label{Fullstabilizing}	
\end{equation}
where $\mathcal{M}$ is a stabilization parameter.

\subsection{Energy dissipation law at the fully discrete level}
Now we set out to prove the energy dissipation law of this full discrete system. Since the entire fully discrete scheme is decoupled into two steps, we update the $\phi$ by solving the CH system in the first step. The free energy of the system is changed as
\begin{equation*}
\frac{F_{h}\left(\phi^{k+1}\right)-F_{h}\left(\phi^{k}\right)}{\Delta t}=\frac{F_{h}\left(\phi^{k+1}, \mathbf{x}^{k}\right)-F_{h}\left(\phi^{k}, \mathbf{x}^{k}\right)}{\Delta t}.  
\end{equation*}
Then, at the second step, particles are moved with the velocity we determined form the NS system. Both the kinetic energy and the free energy are changed.
\begin{equation*}
\frac{E_{k, h}\left(\mathbf{v}^{k+1}\right)-E_{k, h}\left(\mathbf{v}^{k}\right)}{\Delta t}+\frac{F_{h}^{\nabla}\left(\mathbf{x}^{k+1}, \phi^{k+1}\right)-F_{h}^{\nabla}\left(\mathbf{x}^{k}, \phi^{k+1}\right)}{\Delta t}.
\end{equation*}
We give two lemmas to support the proof of the energy dissipation law at the fully discrete level for the CH system and the NS system, respectively.
%
\begin{lem}  \label{lemFullCHenergy}
Assume that the numerical schemes for the CH system are \eqref{FullCH1}, \eqref{FullCH2} and schemes as \eqref{Fullmub}, \eqref{Fullmunabla} are used for chemical potential. The free energy will dissipate with time as follows:
\begin{equation}
\frac{F_{h}\left(\phi^{k+1}\right)-F_{h}\left(\phi^{k}\right)}{\Delta t} \leq M \sum_{i} \sum_{j} \frac{m_{i} m_{j}}{\rho_{i} \rho_{j}}\left(\mu_{i}^{k+1}-\mu_{j}^{k+1}\right)^{2} \omega\left(r_{ij}\right) \leq 0.  \label{FullCHenergy}
\end{equation}
\end{lem}
%
\begin{proof}
\textbf{\textit{For homogeneous part:}}
By introducing an intermediate function $c\left(\phi_{i}\right)=\frac{1}{2}\left(\xi+1-\phi_{i}^{2}\right)$ as \cite{wang2021linear}, we have energy density 
$$f_{b}\left(\phi_{i}\right)=\lambda_{b}\left(c\left(\phi_{i}\right)^{2}+\frac{1}{2} \xi \phi_{i}^{2}-\frac{1}{2} \xi-\frac{1}{4} \xi^{2}\right).$$
A sufficient large $\xi>0$ ensures $\phi_{i} \in I_{\xi}=[-\sqrt{\xi+1}, \sqrt{\xi+1}]$. Obviously, $c\left(\phi_{i}\right)$ is a non-negative concave function and the quadratic function $\phi_{i}^{2}$ is a convex function. We note that $\bar{c}=c\left(\phi_{i}^{k}\right)+c^{\prime}\left(\phi_{i}^{k}\right)\left(\phi_{i}^{k+1}-\phi_{i}^{k}\right)$ and $0 \leq c\left(\phi_{i}^{k+1}\right) \leq \bar{c}$ for $\phi_{i}^{k}, \phi_{i}^{k+1} \in I_{\xi}$. The energy density change can be estimated by
\begin{equation*}
\begin{aligned}
f_{b}\left(\phi_{i}^{k+1}\right)-f_{b}\left(\phi_{i}^{k}\right) &=\lambda_{b}\left[c\left(\phi_{i}^{k+1}\right)^{2}-c\left(\phi_{i}^{k}\right)^{2}+\frac{1}{2} \xi\left(\left(\phi_{i}^{k+1}\right)^{2}-\left(\phi_{i}^{k}\right)^{2}\right)\right] \\
& \leq \lambda_{b}\left[\bar{c}^{2}-c\left(\phi_{i}^{k}\right)^{2}+\xi \phi_{i}^{k+1}\left(\phi_{i}^{k+1}-\phi_{i}^{k}\right)\right] \\
&=\lambda_{b}\left[\left(\bar{c}+c\left(\phi_{i}^{k}\right)\right) c^{\prime}\left(\phi_{i}^{k}\right)+\xi \phi_{i}^{k+1}\right]\left(\phi_{i}^{k+1}-\phi_{i}^{k}\right).
\end{aligned}	
\end{equation*}
By substituting values of $c^{\prime}\left(\phi_{i}^{k}\right)$ and $c\left(\phi_{i}^{k}\right)$, the semi-implicit scheme for $\mu_{i}^{b}$ becomes
\begin{equation*}
\begin{split}
\mu_{i}^{b}\left(\phi_{i}^{k}, \phi_{i}^{k+1}\right)&=\lambda_{b}\left[\left(\bar{c}+c\left(\phi_{i}^{k}\right)\right) c^{\prime}\left(\phi_{i}^{k}\right)+\xi \phi_{i}^{k+1}\right]\\
&=\lambda_{b}\left[\left(\xi+\left(\phi_{i}^{k}\right)^{2}\right) \phi_{i}^{k+1}-(\xi+1) \phi_{i}^{k}\right].
\end{split}
\end{equation*}
It naturally yields
\begin{equation}
f_{b}\left(\phi_{i}^{k+1}\right)-f_{b}\left(\phi_{i}^{k}\right) \leq f_{b}^{\prime}\left(\phi_{i}^{k}, \phi_{i}^{k+1}\right)\left(\phi_{i}^{k+1}-\phi_{i}^{k}\right). \label{Fullyfb}
\end{equation}
Therefore, we conclude that there is always a sufficient large enough $\xi>0$ to make inequality \eqref{Fullyfb} hold. For the entire system
\begin{equation}
\frac{F_{h}^{b}\left(\phi^{k+1}\right)-F_{h}^{b}\left(\phi^{k}\right)}{\Delta t} \leq \sum_{i} \frac{m_{i}}{\rho_{i}}\left(f_{b}^{\prime}\left(\phi_{i}^{k}, \phi_{i}^{k+1}\right), \frac{\phi_{i}^{k+1}-\phi_{i}^{k}}{\Delta t}\right).
\label{FullFbineq}
\end{equation}
\begin{rem}
This kind of energy-factorization method finally results in a linear energy-stable discrete scheme for $\mu_i^{b}$, which benefits the implementation. It must be noted that other well-known schemes like SAV and IEQ can be applied to this part as well. Also, the numerical value of $\phi_i$ can be bounded, which validates our conclusion. The relevant proof can be referred to \cite{li2016characterizing,li2017second,li2021convergence}.
\end{rem}

\textbf{\textit{For inhomogeneous part:}} We turn to the scheme of chemical potential $\mu_{\nabla}$ caused by inhomogeneity. Since the corresponding energy $F_{\nabla}$ is not a local function, it depends not only on the $\phi$ value but also on the position $\mathbf{x}$ of all particles. In the current step for solving the CH system, positions of the particles will not change. An intermediate variable $F_{i, h}^{\nabla}$ for a single particle is defined for our proof as follows,
\begin{equation}
F_{i, h}^{\nabla}=\frac{\lambda}{2}\left(\nabla_{h} \phi_{i}\right) \cdot\left(\nabla_{h} \phi_{i}\right).
\end{equation}
The free energy caused by the inhomogeneity of system$F_{h}^{\nabla}=\sum_{i} \frac{m_{i}}{\rho_{i}} F_{i, h}^{\nabla}$.Since $F_{i, h}^{\nabla}$ is not a local function as well, the convex-concave property is related to its Hessian matrix. Taking the first order derivative with respect to,
\begin{equation*}
\frac{\partial F_{i, h}^{\nabla}}{\partial \phi_{m}}
= \lambda \sum_{j} \frac{m_{i}}{\rho_{j}}\left(\delta_{i m}-\delta_{j m}\right) \nabla_{i} W_{i j} \sum_{k} \frac{m_{i}}{\rho_{k}}\left(\phi_{i}-\phi_{k}\right) \nabla_{i} W_{i k}.
\end{equation*}
Then, taking the derivative with respect to $\phi_n$, we have elements of its Hessian matrix,
\begin{equation*}
\frac{\partial^{2} F_{i, h}^{\nabla}}{\partial \phi_{m} \phi_{n}}= \lambda \sum_{m} \sum_{j} \frac{m_{i}}{\rho_{j}}\left(\delta_{i m}-\delta_{j m}\right) \nabla_{i} W_{i j} \sum_{n} \sum_{k} \frac{m_{i}}{\rho_{k}}\left(\delta_{i n}-\delta_{k n}\right) \nabla_{i} W_{i k}.	 \label{Fnabla1stde}
\end{equation*}
For any vector $\bold{b} \neq \bold{0}$, we have
\begin{equation}
\sum_{m} \sum_{n} b_{m} \frac{\partial^{2} F_{i, h}^{\nabla}}{\partial \phi_{m} \phi_{n}} b_{n}
=  \lambda \left(\sum_{n} \sum_{k} b_{n} \frac{m_{i}}{\rho_{k}}\left(\delta_{i n}-\delta_{k n}\right) \nabla_{i} W_{i k}\right)^{2} \geq 0. \label{Fnabla2ndde}
\end{equation}
The Hessian matrix $\frac{\partial^{2} F_{i, h}^{\nabla}}{\partial \phi_{m} \partial \phi_{n}}$ is positive-definition. Thus, $F_{i, h}^{\nabla}$ is convex with respect to $\phi$. If the chemical potential $\mu_{i}^{\nabla}$ is treated with a fully implicit scheme as \eqref{Fullmunabla}, it will make the energy inequality hold as
\begin{equation}
\frac{F_{h}^{\nabla}\left(\phi^{k+1}\right)-F_{h}^{\nabla}\left(\phi^{k}\right)}{\Delta t} \leq \sum_{i} \frac{m_{i}}{\rho_{i}} \sum_{j}\left(\frac{\partial F_{i, h}^{\nabla}\left(\phi^{k+1}\right)}{\partial \phi_{j}^{k+1}}, \frac{\phi_{j}^{k+1}-\phi_{j}^{k}}{\Delta t} \right).
\end{equation}
Since we consider the incompressible fluid, it yields the following identity:
\begin{equation*}
\sum_{i} \frac{m_{i}}{\rho_{i}}\left(\frac{\rho_{i}}{m_{i}} \frac{\partial F_{h}^{\nabla}\left(\phi^{k+1}\right)}{\partial \phi_{i}^{k+1}}, \frac{\phi_{i}^{k+1}-\phi_{i}^{k}}{\Delta t}\right)=\sum_{i} \frac{m_{i}}{\rho_{i}} \sum_{j}\left(\frac{\partial F_{i, h}^{\nabla}\left(\phi^{k+1}\right)}{\partial \phi_{j}^{k+1}}, \frac{\phi_{j}^{k+1}-\phi_{j}^{k}}{\Delta t}\right).	
\end{equation*}
Then we find the inequality
\begin{equation}
\frac{F_{h}^{\nabla}\left(\phi^{k+1}\right)-F_{h}^{\nabla}\left(\phi^{k}\right)}{\Delta t} \leq \sum_{i} \frac{m_{i}}{\rho_{i}}\left(\frac{\rho_{i}}{m_{i}} \frac{\partial F_{h}^{\nabla}\left(\phi^{k+1}\right)}{\partial \phi_{i}^{k+1}}, \frac{\phi_{i}^{k+1}-\phi_{i}^{k}}{\Delta t}\right). \label{FullFnabineq}
\end{equation}
Similarly, we use the inner product to get the energy dissipation law at the fully discrete level. By taking the inner product of equation \eqref{FullCH1} with $\mu_{i}^{k+1}$
\begin{equation}
\sum_{i} \frac{m_{i}}{\rho_{i}}\left(\frac{\phi_{i}^{k+1}-\phi_{i}^{k}}{\Delta t}, \mu_{i}^{k+1} \right) = M \sum_{i} \sum_{j} \frac{m_{i} m_{j}}{\rho_{i} \rho_{j}}\left(\mu_{i}^{k+1}-\mu_{j}^{k+1}\right)^{2} \omega\left(r_{i j}\right).
\label{FullCH1energy}
\end{equation}
By taking the inner product of \eqref{FullCH2} with $\frac{\phi_{i}^{k+1}-\phi_{i}^{k}}{\Delta t}$, we have
\begin{equation}
\begin{split}
\sum_{i} \frac{m_{i}}{\rho_{i}}\left( \mu_{i}^{k+1}, \frac{\phi_{i}^{k+1}-\phi_{i}^{k}}{\Delta t}\right)
= &\sum_{i} \frac{m_{i}}{\rho_{i}} \left( \frac{\rho_{i}}{m_{i}} \frac{\partial F_h^{\nabla}\left(\phi^{k+1}\right)}{\partial \phi_{i}^{k+1}}, \frac{\phi_{i}^{k+1}-\phi_{i}^{k}}{\Delta t}\right)\\
  &+\sum_{i} \frac{m_{i}}{\rho_{i}}\left( f_{b}^{\prime}\left(\phi_{i}^{k}, \phi_{i}^{k+1}\right), \frac{\phi_{i}^{k+1}-\phi_{i}^{k}}{\Delta t}\right).  \label{FullCH2energy}
\end{split}
\end{equation}
Owing to $F_{h}\left(\phi\right) = F_{h}^{b}\left(\phi\right) + F_{h}^{\nabla}\left(\phi\right)$, with the above equations \eqref{FullCH1energy}, \eqref{FullCH2energy} and energy inequalities \eqref{FullFbineq}, \eqref{FullFnabineq} derived from the homogeneous and inhomogeneous parts, we get the energy dissipation law of the CH system as \eqref{FullCHenergy}
\end{proof}
%
%
%
%
\begin{lem} \label{lemFullNSenergy}
Assume the fully discrete scheme for NS system as \eqref{Fullmomentum}, \eqref{FullPo}, \eqref{Fulldiv}, \eqref{FullUpdate} is used. Among them, a stabilizing semi-implicit scheme is applied to the interfacial term as \eqref{Fullstabilizing}. Then a sufficient great $\mathcal{M}$  always makes the energy dissipation law of the NS system at the full discrete level hold well.
\begin{equation}
\frac{E_{k, h}\left(\mathbf{v}^{k+1}\right)-E_{k, h}\left(\mathbf{v}^{k}\right)}{\Delta t} +
\frac{F_{h}^{\nabla}\left(\mathbf{x}^{k+1}, \phi^{k+1}\right)-F_{h}^{\nabla}\left(\mathbf{x}^{k}, \phi^{k+1}\right)}{\Delta t} \leq 0.  \label{FullNSenergy}
\end{equation}
\end{lem}
\begin{proof}
It has been mentioned that the change of energy $F^{\nabla}_h$ is partially caused by the change of the position of particles $\mathbf{x}$. When the position of particles is updated through \eqref{FullUpdate} and $\phi$ is unvarying, this part of energy change characterizes the energy transfer from kinetic energy $E_k$ to the free energy $F$. This relationship is consistent with the definition in the ODE system. Using the symbol $\langle\rangle_{h}^{t e m}$ to denote the fully discrete scheme, including temporal discretization. Thus,
\begin{equation}
\langle\text {transfer}_\mathrm{EK-F} \rangle_{h}^{tem} = \frac{F_{h}^{\nabla}\left(\mathbf{x}^{k+1}, \phi^{k+1}\right)-F_{h}^{\nabla}\left(\mathbf{x}^{k}, \phi^{k+1}\right)}{\Delta t}. \label{FulldefEK-F}
\end{equation}
Correspondingly, the energy transferred from free energy $F$ to kinetic energy $E_k$ is characterized by the inner product of interfacial term  with $m_i\tilde{\mathbf{v}}_{i}^{k+1}$.
\begin{equation}
\begin{split}
\left\langle\operatorname{transfer}_{\mathrm{F}-\mathrm{EK}}\right\rangle_{h}^{\mathrm{tem}}
&= -\sum_{i}\left(m_{i} \tilde{\mathbf{v}}_{i}^{k+1} \cdot \frac{\left\langle\mathcal{T}_{i}\left(\tilde{\mathbf{v}}^{k+1}, \mathbf{v}^{k}\right)\right\rangle_{h}}{m_{i}}\right)\\
&= -\sum_{i} \tilde{\mathbf{v}}_{i}^{k+1} \cdot\left\langle\mathcal{T}_{i}\left(\tilde{\mathbf{x}}^{k+1}, \mathbf{x}^{k}\right)\right\rangle_{h} 	
  \approx -\sum_{i}\left\langle\frac{\partial F_{h}^{\nabla}}{\partial \mathbf{x}_{i}} \frac{\partial \mathbf{x}_{i}}{\partial t}\right\rangle_{h}. \label{FulldefF-EK}
\end{split}
\end{equation}
Since the NS system is determined by projection methods, we split the entire process into three steps: (1) solving the momentum equation; (2) solving the Poisson equation; (3) particle movement. So the energy change can be expressed as the summation of

$\frac{\left[E_{k, h}\left(\tilde{\mathbf{v}}^{k+1}\right)-E_{k, h}\left(\mathbf{v}^{k}\right)\right]}{\Delta t}, \frac{\left[E_{k, h}\left(\mathbf{v}^{k+1}\right)-E_{k, h}\left(\tilde{\mathbf{v}}^{k+1}\right)\right]}{\Delta t}$, and $\frac{\left[F_{h}^{\nabla}\left(\mathbf{x}^{k+1}, \phi^{k+1}\right)-F_{h}^{\nabla}\left(\mathbf{x}^{k}, \phi^{k+1}\right)\right]}{\Delta t}$.

With the help of inner product operation to get the energy and based on the inequality $\frac{1}{2}\|\mathbf{a}\|^{2}-\frac{1}{2}\|\mathbf{b}\|^{2} \leq \mathbf{a} \cdot(\mathbf{a}-\mathbf{b})$, we can get the energy relations for each step.

\textbf{\textit{(1) momentum equation: }}
By multiplying $m_{i} \tilde{\mathbf{v}}_{i}^{k+1}$ with the momentum equation \eqref{Fullmomentum}, the left-hand side (LHS) becomes
\begin{equation*}
E_{k, h}\left(\tilde{\mathbf{v}}^{k+1}\right)-E_{k, h}\left(\mathbf{v}^{k}\right)=\frac{1}{2} \sum_{i} m_{i}\left\|\tilde{\mathbf{v}}_{i}^{k+1}\right\|^{2}-\frac{1}{2} \sum_{i} m_{i}\left\|\mathbf{v}_{i}^{k}\right\|^{2} \leq \sum_{i} m_{i} \tilde{\mathbf{v}}_{i}^{k+1} \cdot\left(\tilde{\mathbf{v}}_{i}^{k+1}-\mathbf{v}_{i}^{k}\right).
\end{equation*}
Then, combining \eqref{FulldefF-EK}, it can concluded that
\begin{equation}
\begin{split}
\frac{E_{k, h}\left(\tilde{\mathbf{v}}^{k+1}\right)-E_{k, h}\left(\mathbf{v}^{k}\right)}{\Delta t} &\leq
\frac{(d+2)}{2} \sum_{i} \sum_{j} \frac{m_{i} m_{j}}{\rho_{i}^{k+1} \rho_{j}^{k+1}} \frac{\left(\eta_{i}+\eta_{j}\right)\left\|\tilde{\mathbf{v}}_{i j}^{k+1} \cdot \mathbf{x}_{i j}^{k}\right\|^{2}}{\left\|\mathbf{x}_{i j}^{k}\right\|^{2}} \omega\left(r_{i j}\right)\\
&+\langle\text {transfer}_\mathrm{F-EK}\rangle_{h}^{tem}. \label{Full1}
\end{split}
\end{equation}

\textbf{\textit{(2) Poisson equation: }}
By multiplying $m_{i}\mathbf{v}_{i}^{k+1}$ with the Poisson equation \eqref{FullPo}, we deduce
\begin{equation}
\sum_{i}\left(m_{i} \mathbf{v}_{i}^{k+1}, \frac{\mathbf{v}_{i}^{k+1}-\tilde{\mathbf{v}}_{i}^{k+1}}{\Delta t}\right)=\sum_{i}\left(m_{i} \mathbf{v}_{i}^{k+1},-\sum_{j} m_{j}\left(\frac{p_{i}^{k+1}}{\left(\rho_{i}^{k+1}\right)^{2}}+\frac{p_{j}^{k+1}}{\left(\rho_{j}^{k+1}\right)^{2}}\right) \nabla_{i} W_{i j}^{n}\right). \label{POinner}
\end{equation}
Based on the divergence-free condition at the discrete level, it naturally yields
\begin{equation*}
-\sum_{i} m_{i} \frac{p_{i}^{k+1}}{\left(\rho_{i}^{k+1}\right)^{2}} \sum_{j} m_{j}\left(\mathbf{v}_{i}^{k+1}-\mathbf{v}_{j}^{k+1}\right) \cdot \nabla_{i} W_{i j}=0.
\end{equation*}
Then, we can conclude the right-hand side (RHS) of \eqref{POinner}
\begin{equation*}
\text { RHS }=-\frac{1}{2} \sum_{i} \sum_{i} m_{i} m_{j}\left(\frac{p_{i}^{k+1}}{\left(\rho_{i}^{k+1}\right)^{2}}+\frac{p_{j}^{k+1}}{\left(\rho_{j}^{k+1}\right)^{2}}\right)\left(\mathbf{v}_{i}^{k+1}-\mathbf{v}_{j}^{k+1}\right) \cdot \nabla_{i} W_{i j}=0.
\end{equation*}
Finally, we get
\begin{equation}
\frac{E_{k, h}\left(\mathbf{v}^{k+1}\right)-E_{k, h}\left(\tilde{\mathbf{v}}_{i}^{k+1}\right)}{\Delta t} \leq \sum_{i}\left(m_{i} \mathbf{v}_{i}^{k+1}, \frac{\mathbf{v}_{i}^{k+1}-\tilde{\mathbf{v}}_{i}^{k+1}}{\Delta t}\right)=0.
\label{Full2}
\end{equation}

\textbf{\textit{(3) movement of particles: }}
Based on the definition of interfacial force: $\boldsymbol{\mathcal{T}}_{i}:=\frac{\partial F_{h}^{\nabla}}{\partial \mathbf{x}_{i}}$ and the definition of $F_{i, h}^{\nabla}$, for $\mathbf{x}_{i}^{k+1}$, $\mathbf{x}_{i}^{k} \in S$ and the open set $S \subseteq \mathbb{R}^{d}(d=2,3)$, let $\mathbf{x}_{i} = (1- \alpha_i)\mathbf{x}_{i}^{k} +\alpha_i \mathbf{x}_{i}^{k+1}$, $\alpha_i \in [0,1]$, and $\mathbf{x}$ is the collection of $\mathbf{x}_{i}$. If we consider the movement as a particle-wise process and based on the mean value theorem for multivariable, we have
\begin{equation*}
F_{h}^{\nabla}\left(\mathbf{x}^{k+1}, \phi^{k+1}\right)-F_{h}^{\nabla}\left(\mathbf{x}^{k}, \phi^{k+1}\right)=\sum_{i} \boldsymbol{\mathcal{T}}_{i}(\mathbf{x}) \cdot\left(\mathbf{x}_{i}^{k+1}-\mathbf{x}_{i}^{k}\right).
\end{equation*}
Through the convex-concave splitting, a semi-implicit scheme can make the following inequality hold:
\begin{equation*}
\boldsymbol{\mathcal{T}}_{i}(\mathbf{x}) \cdot\left(\mathbf{x}_{i}^{k+1}-\mathbf{x}_{i}^{k}\right) \leq\left\langle \boldsymbol{\mathcal{T}}_{i}\left(\mathbf{x}^{k+1}, \mathbf{x}^{k}\right)\right\rangle_{h} \cdot\left(\mathbf{x}_{i}^{k+1}-\mathbf{x}_{i}^{k}\right).
\end{equation*}
namely,
\begin{equation}
\langle\text {transfer}_\mathrm{EK-F} \rangle_{h}^{tem}
=\frac{F_{h}^{\nabla}\left(\mathbf{x}^{k+1}, \phi^{k+1}\right)-F_{h}^{\nabla}\left(\mathbf{x}^{k}, \phi^{k+1}\right)}{\Delta t}
\leq
\sum_{i} \mathbf{v}_{i}^{k+1} \cdot\left\langle \boldsymbol{\mathcal{T}}_{i}\left(\mathbf{x}^{k+1}, \mathbf{x}^{k}\right)\right\rangle_{h}.  \label{ineqseminext}
\end{equation}
By employing a stabilizing term, the energy can be rewritten in a modified form: $F_{h}^{\nabla}=F_{h, e}^{\nabla}+F_{h, c}^{\nabla}$, where subscript $e$ means convex and $c$ means concave. Then the original energy $F_{h}^{\nabla}$ can be split into a convex part and a concave part:
\begin{equation*}
\left\{\begin{array}{l}
F_{h, e}^{\nabla}= \frac{\mathcal{M}}{\Delta t} \sum_{i}\left\|\mathbf{x}_{i}\right\|^{2}, \\
F_{h, c}^{\nabla}=F_{h}^{\nabla}-\frac{\mathcal{M}}{\Delta t} \sum_{i}\left\|\mathbf{x}_{i}\right\|^{2}.
\end{array}\right.	
\end{equation*}
Their second order derivatives are
\begin{equation*}
\frac{\partial^{2} F_{h, e}^{\nabla}}{\partial \mathbf{x}_{i} \partial \mathbf{x}_{j}}=\left\{\begin{array}{ll}
2 \frac{\mathcal{M}}{\Delta t} & , i=j, \\
0 \quad &,  i \neq j,
\end{array} \quad \text { and } \quad \frac{\partial^{2} F_{h, c}^{\nabla}}{\partial \mathbf{x}_{i} \partial \mathbf{x}_{j}}= \begin{cases}-2 \frac{\mathcal{M}}{\Delta t}+\frac{\partial \mathcal{T}_{i}}{\partial \mathbf{x}_{i}} & , i=j, \\
\frac{\partial \mathcal{T}_{i}}{\partial \mathbf{x}_{j}} & , i \neq j.\end{cases}\right.
\end{equation*}
The Hessian matrix of $F_{h, e}^{\nabla}$ is always positive-definite. Moreover, a sufficiently large modified parameter $\mathcal{M}$ always exists to make the Hessian matrix of $F_{h, c}^{\nabla}$ negative-definite. Based on the convex-concave properties, this splitting makes the following inequalities hold:
\begin{equation} \label{convex-concave}
\left\{\begin{array}{l}
F_{h, e}^{\nabla}\left(\mathbf{x}^{k+1}, \phi^{k+1}\right)-F_{h, e}^{\nabla}\left(\mathbf{x}^{k}, \phi^{k+1}\right) \leq \sum_{i} 2 \mathcal{M} \mathbf{x}_{i}^{k+1} \cdot\left(\mathbf{x}_{i}^{k+1}-\mathbf{x}_{i}^{k}\right), \\
F_{h, c}^{\nabla}\left(\mathbf{x}^{k+1}, \phi^{k+1}\right)-F_{h, c}^{\nabla}\left(\mathbf{x}^{k}, \phi^{k+1}\right) \leq \sum_{i}\left(\mathcal{T}_{i}\left(\mathbf{x}^{k}\right)-2 \mathcal{M} \mathbf{x}_{i}^{k}\right) \cdot\left(\mathbf{x}_{i}^{k+1}-\mathbf{x}_{i}^{k}\right).
\end{array}\right.	
\end{equation}
Combining these two inequalities \eqref{convex-concave}, we find the appropriate semi-implicit fully discrete scheme for interfacial force
\begin{equation}
\left\langle \boldsymbol{\mathcal{T}}_{i}\left(\mathbf{x}^{k+1}, \mathbf{x}^{k}\right)\right\rangle_{h} =\frac{2\mathcal{M}}{\Delta t} \mathbf{x}_{i}^{k+1}+\boldsymbol{\mathcal{T}}_{i}\left(\mathbf{x}^{k}\right)- \frac{2\mathcal{M}}{\Delta t} \mathbf{x}_{i}^{k} = 2\mathcal{M} \mathbf{v}_{i}^{k+1} +\boldsymbol{\mathcal{T}}_{i}\left(\mathbf{x}^{k}\right). \label{semi_T}
\end{equation}
The scheme for $\boldsymbol{\mathcal{T}}_{i}$ as \eqref{semi_T}  will ensure inequality \eqref{ineqseminext}. For the physical consistency of the full discrete scheme, we apply the same scheme for the term $\left\langle \boldsymbol{\mathcal{T}}_{i}\left(\tilde{\mathbf{x}}^{k+1}, \mathbf{x}^{k}\right)\right\rangle_{h}$ with the intermediate velocity $\tilde{\mathbf{v}}_{i}^{k+1}$ in the momentum equation as \eqref{Fullstabilizing}. Thus,
\begin{equation*}
\left\langle\text{transfer}_\mathrm{F-EK}\right\rangle_{h}^{tem}
= -\sum_{i} \tilde{\mathbf{v}}_{i}^{k+1} \cdot\left[2 M \tilde{\mathbf{v}}_{i}^{k+1}+\boldsymbol{\mathcal{T}}_{i}\left(\mathbf{x}^{k}\right)\right].
\end{equation*}
By recalling the inequality \eqref{Full2} derived from the pressure Poisson equation, we can conclude that there is always a sufficiently large coefficient $\mathcal{M}$ to make
\begin{equation}
\begin{aligned}
& \quad \langle\text {transfer}_\mathrm{EK-F} \rangle_{h}^{tem} + \langle\text {transfer}_\mathrm{F-EK} \rangle_{h}^{tem} \\
&\leq \sum_{i} \mathbf{v}_{i}^{k+1}  \cdot\left[2 \mathcal{M} \mathbf{v}_{i}^{k+1}  +\mathcal{T}_{i}\left(\mathbf{x}^{k}\right)\right] - \sum_{i} \tilde{\mathbf{v}}_{i}^{k+1}  \cdot\left[2 \mathcal{M} \tilde{\mathbf{v}}_{i}^{k+1} +\mathcal{T}_{i}\left(\mathbf{x}^{k}\right)\right] \\
& = 2\mathcal{M}  \sum_{i}\left[\left(\mathbf{v}_{i}^{k+1}\right)^{2}-\left(\tilde{\mathbf{v}}_{i}^{k+1}\right)^{2}\right]
+ \sum_{i}\left(\mathbf{v}_{i}^{k+1}-\tilde{\mathbf{v}}_{i}^{k+1}\right) \mathcal{T}_{i}\left(\mathbf{x}^{k}\right) \leq 0.
\end{aligned}     \label{Full3}
\end{equation}
At last, combining the energy changes at three steps as a
\eqref{Full1}, \eqref{Full2}, and \eqref{Full3}, we reach the conclusion of the lemma as \eqref{FullNSenergy}.
\end{proof}
\begin{rem}
We could consider the conservation of energy transfer still maintained at the fully discrete level. But the energy inequality of \eqref{Full3} is equivalent to adding a numerical dissipative term to help maintain the energy decaying property. It comes from the numerical approximation of energy transfer.
\begin{equation}
\left\langle\operatorname{transfer}_{diss}\right\rangle_{h}^{\mathrm{tem}} = 2 \mathcal{M} \sum_{i}\left[\left(\mathbf{v}_{i}^{k+1}\right)^{2}-\left(\tilde{\mathbf{v}}_{i}^{k+1}\right)^{2}\right]+\sum_{i}\left(\mathbf{v}_{i}^{k+1}-\tilde{\mathbf{v}}_{i}^{k+1}\right) \mathcal{T}_{i}\left(\mathbf{x}^{k}\right) \leq 0. \label{FullDiss}
\end{equation}
The value of $\mathcal{M}$ is adjusted adaptively posteriori based on the solved $\mathbf{v}_{i}^{k+1}$. If the above energy inequality \eqref{FullDiss} does not hold, the value of $\mathcal{M}$ can be appropriately increased to maintain the energy decaying trend.
\end{rem}

Finally, the theorem for the complete NSCH system at the fully discrete level is proposed here.
\begin{thm} \label{Fullenergydiss}
Assume the numerical scheme for the NSCH system based on the SPH method includes \eqref{FullCH1}, \eqref{FullCH2}, \eqref{Fullrho}, \eqref{Fullmomentum}, \eqref{FullPo}, \eqref{Fulldiv} and \eqref{FullUpdate}. The energy dissipation law holds at the fully discrete level.
\begin{equation}
\frac{E_{\text {total}, h}^{k+1}-E_{\text {total}, h}^{k}}{\Delta t}=\frac{E_{k, h}\left(\mathbf{v}^{k+1}\right)-E_{k, h}\left(\mathbf{v}^{k}\right)}{\Delta t}+\frac{F_{h}\left(\mathbf{x}^{k+1}, \phi^{k+1}\right)-F_{h}\left(\mathbf{x}^{k}, \phi^{k}\right)}{\Delta t} \leq 0.    \label{Fullenergyall}
\end{equation}
\end{thm}
\begin{proof}
This energy inequality for the NSCH system can be concluded through adding these two inequalities, \eqref{FullCHenergy} and \eqref{FullNSenergy}, from Lemmas \ref{lemFullCHenergy} and \ref{lemFullNSenergy} for the CH and NS systems, respectively.
\end{proof}

\section{Numerical examples} \label{section:Numerical examples}
\renewcommand{\thefigure}{\arabic{section}.\arabic{figure}}
In this section, we present some numerical results to show the performance of the proposed energy-stable SPH method for the NSCH two-phase system. We consider the evolution of a droplets with the interfacial tension. In all of the numerical cases, the spatial domain is taken as $\Omega = [-1,1]^2$. The particles are placed uniformly in the initial step, and there are $40\times40$ particles in all examples. The smoothing length is set as $h$.

\subsection{The dynamics of a square droplet}
In the first numerical example, we consider the evolution process of a square droplet in a domain. The parameters are chosen as:
$h = 0.05, \Delta t = 0.01, T = 10, \mathcal{M} = 10^3, \eta = 1, \lambda = 1, \epsilon=0.02, \text{Mobility}=0.002, \xi = 1.$
The initial phase variables are chosen as $ \phi = 1$ inside the square droplet and $\phi = -1$ everywhere else.
\begin{figure}[!b]
    \centering \subfigure[Initial]{
    \begin{minipage}[b]{0.3\textwidth}
    \centering
    \includegraphics[width=1.0\textwidth,height=1.4in]{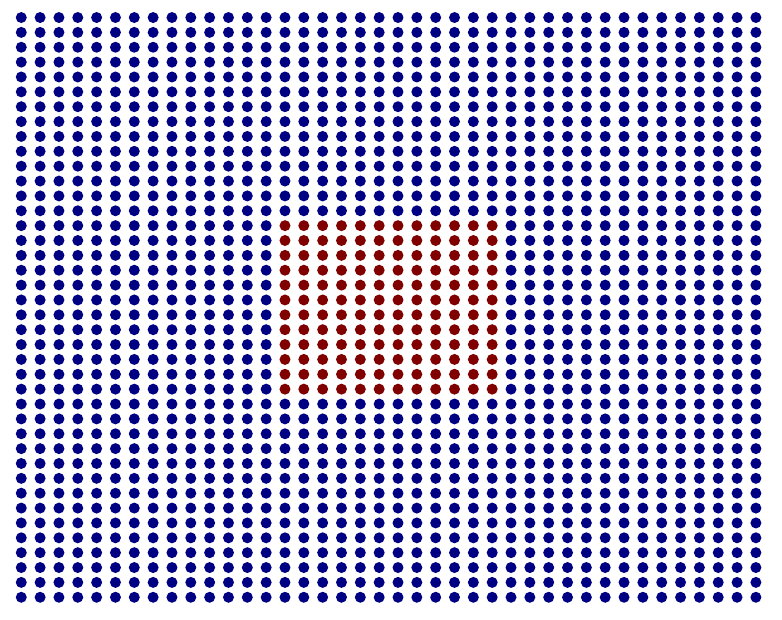}
    \end{minipage}
    }
    \centering \subfigure[$t=1$]{
    \begin{minipage}[b]{0.3\textwidth}
    \centering
    \includegraphics[width=1.0\textwidth,height=1.4in]{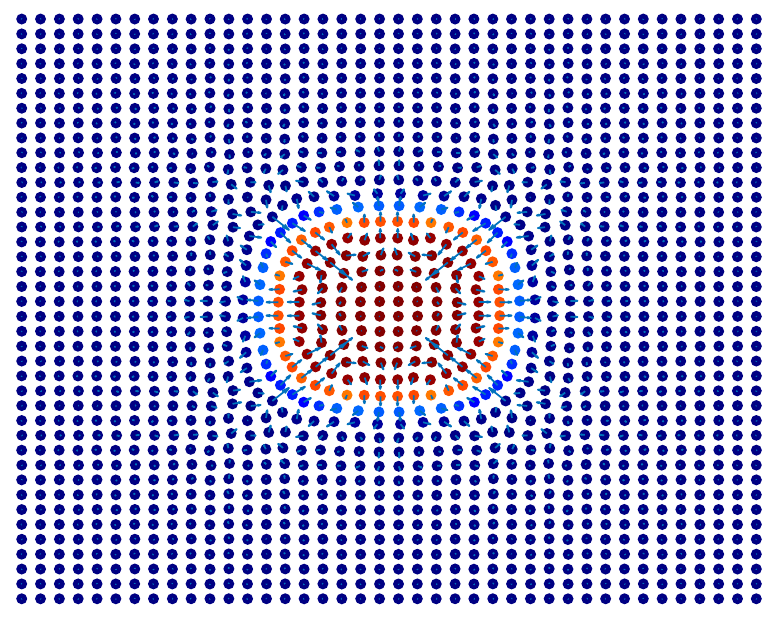}
    \end{minipage}
    }
    \centering \subfigure[$t=20$]{
    \begin{minipage}[b]{0.3\textwidth}
    \centering
    \includegraphics[width=1.0\textwidth,height=1.4in]{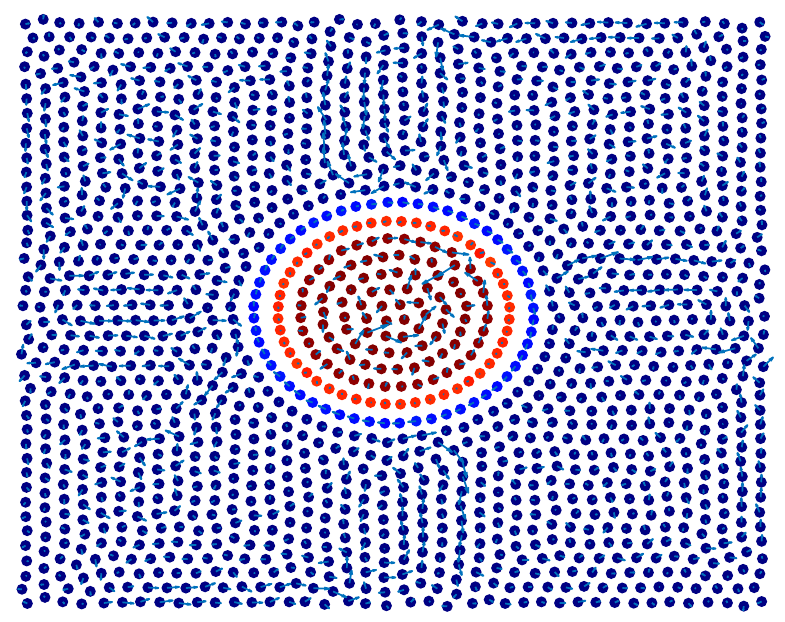}
    \end{minipage}
    }
    \caption{The evolution of $\phi$ and spatial particle distribution with time in example 1}
    \label{example 1_particles}
\end{figure}
 \begin{figure}[!b]
    \centering \subfigure[Initial]{
    \begin{minipage}[b]{0.3\textwidth}
    \centering
    \includegraphics[width=1.0\textwidth,height=1.4in]{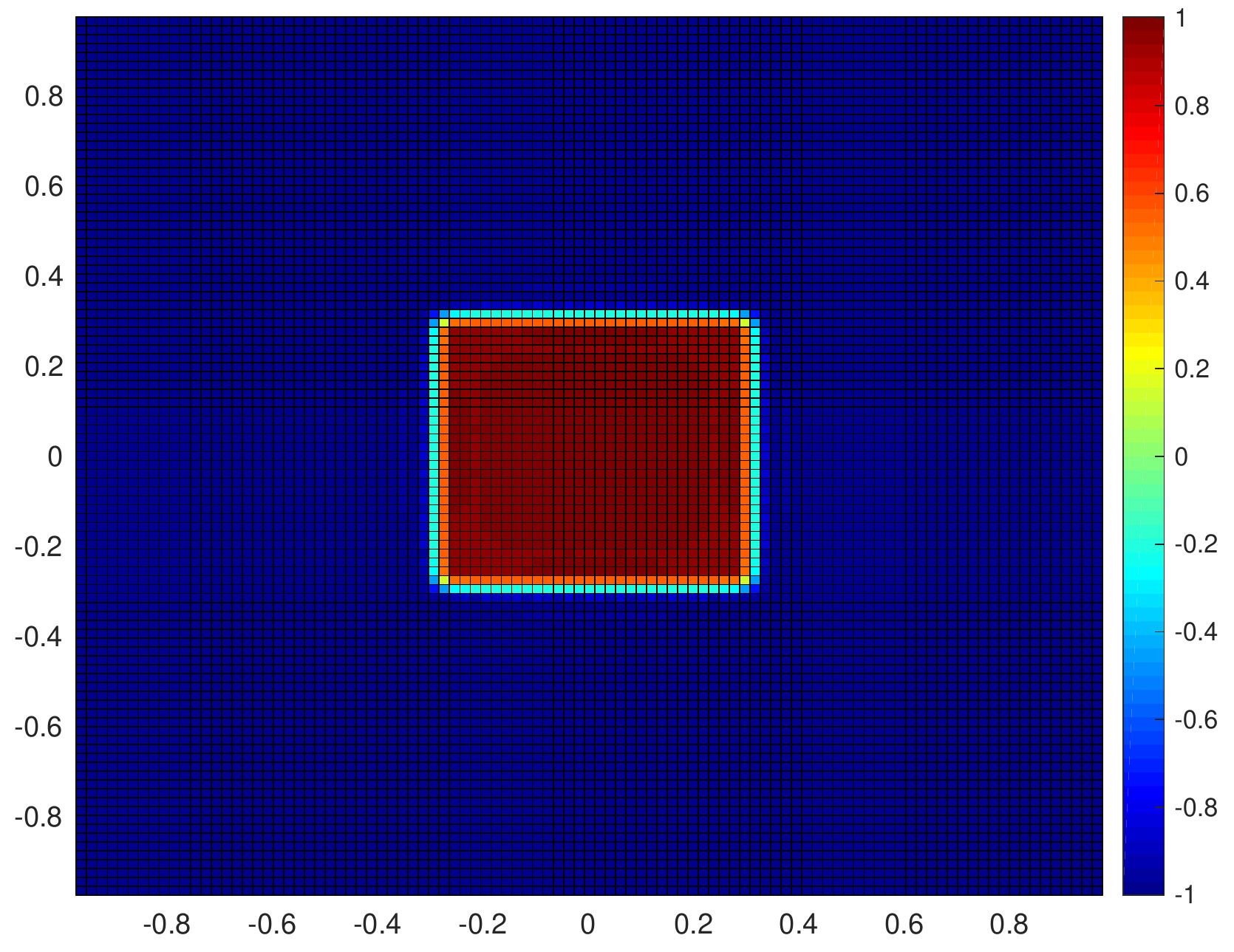}
    \end{minipage}
    }
    \centering \subfigure[$t=1$]{
    \begin{minipage}[b]{0.3\textwidth}
    \centering
    \includegraphics[width=1.0\textwidth,height=1.4in]{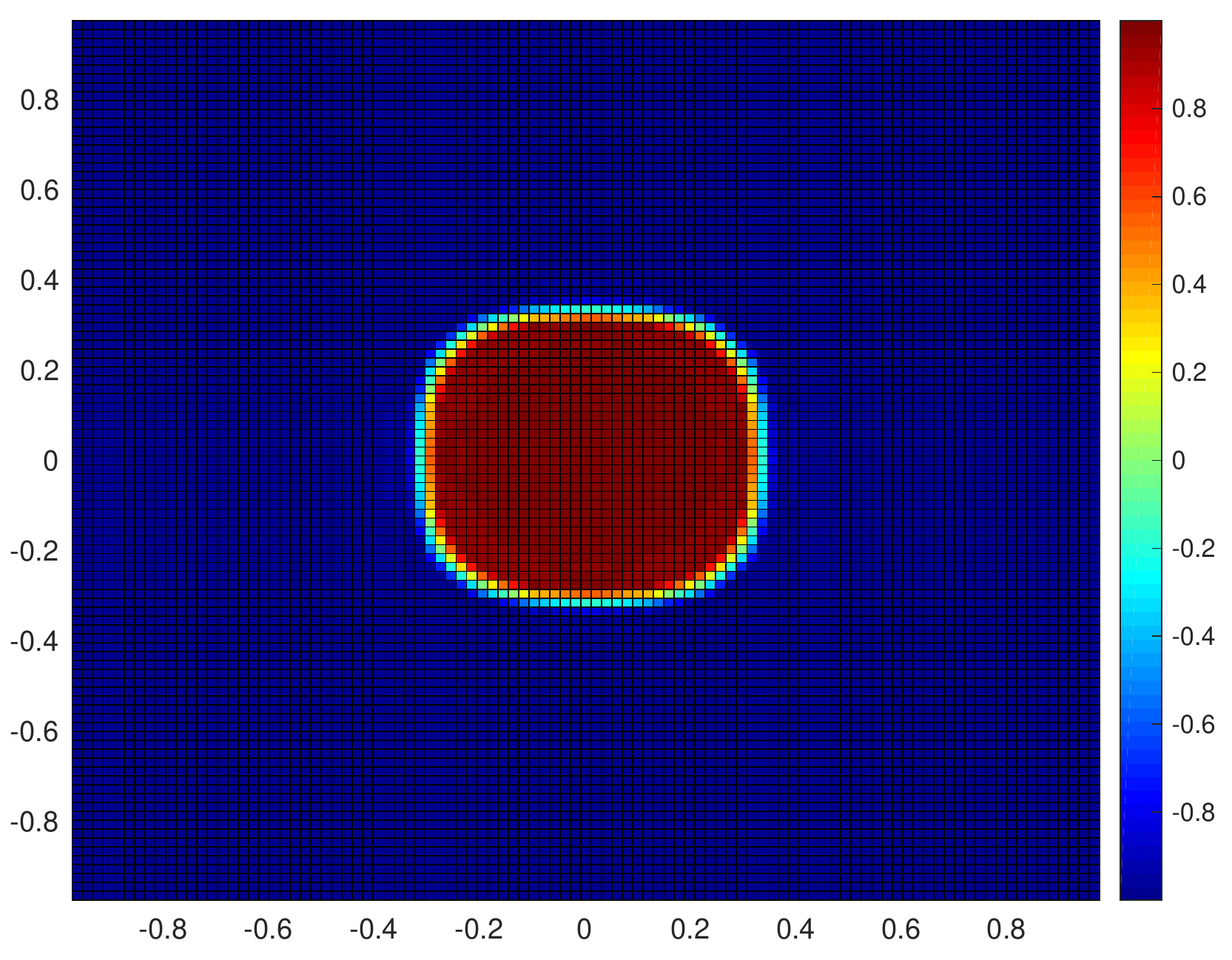}
    \end{minipage}
    }
    \centering \subfigure[$t=20$]{
    \begin{minipage}[b]{0.3\textwidth}
    \centering
    \includegraphics[width=1.0\textwidth,height=1.4in]{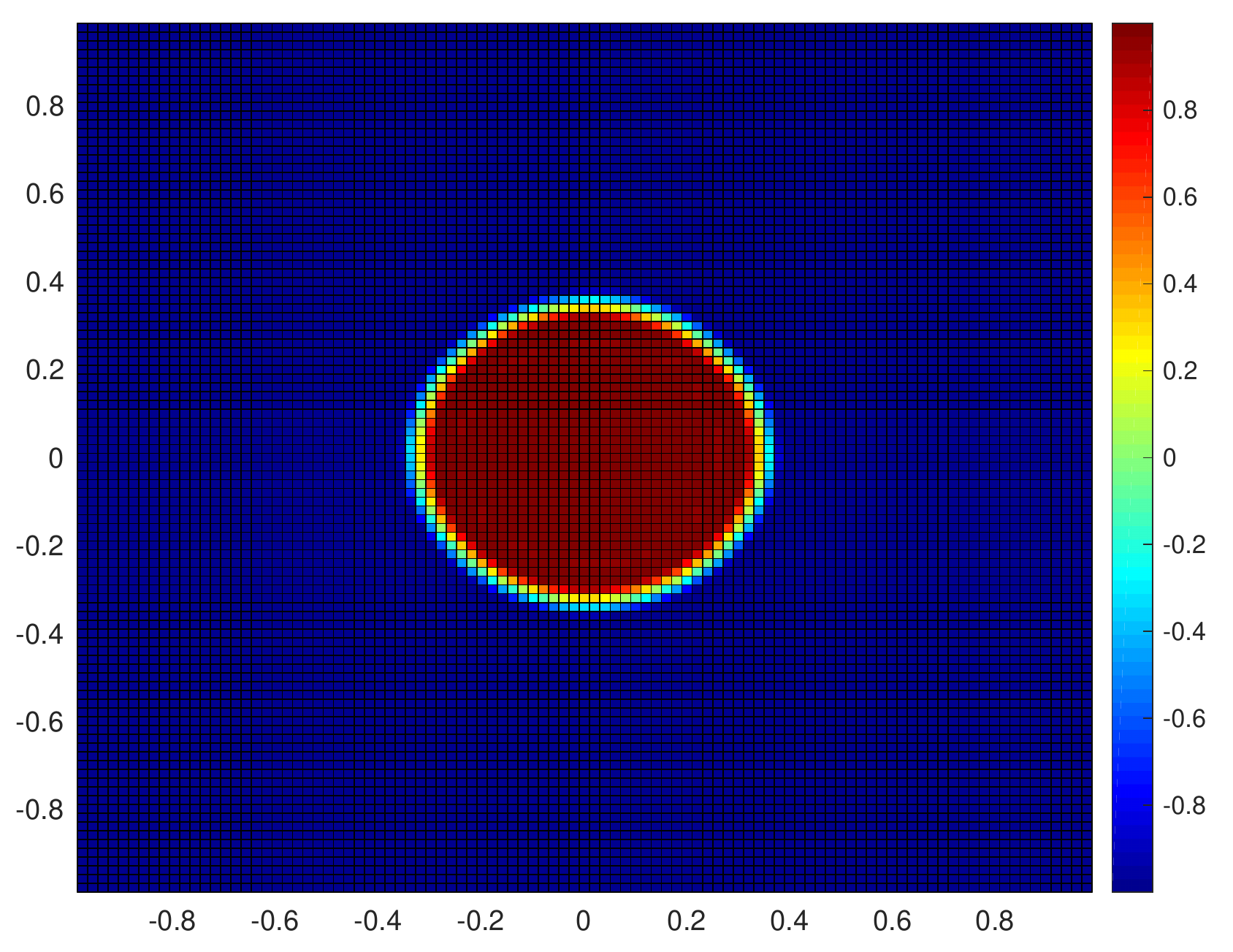}
    \end{minipage}
    }
    \caption{The contour of $\phi$ interpolated by the SPH operator}
    \label{example 1_contour}
 \end{figure}
 \begin{figure}[!t]
    \centering \subfigure[Energy]{
    \begin{minipage}[b]{0.45\textwidth}
    \centering
    \includegraphics[width=0.9\textwidth,height=0.75\textwidth]{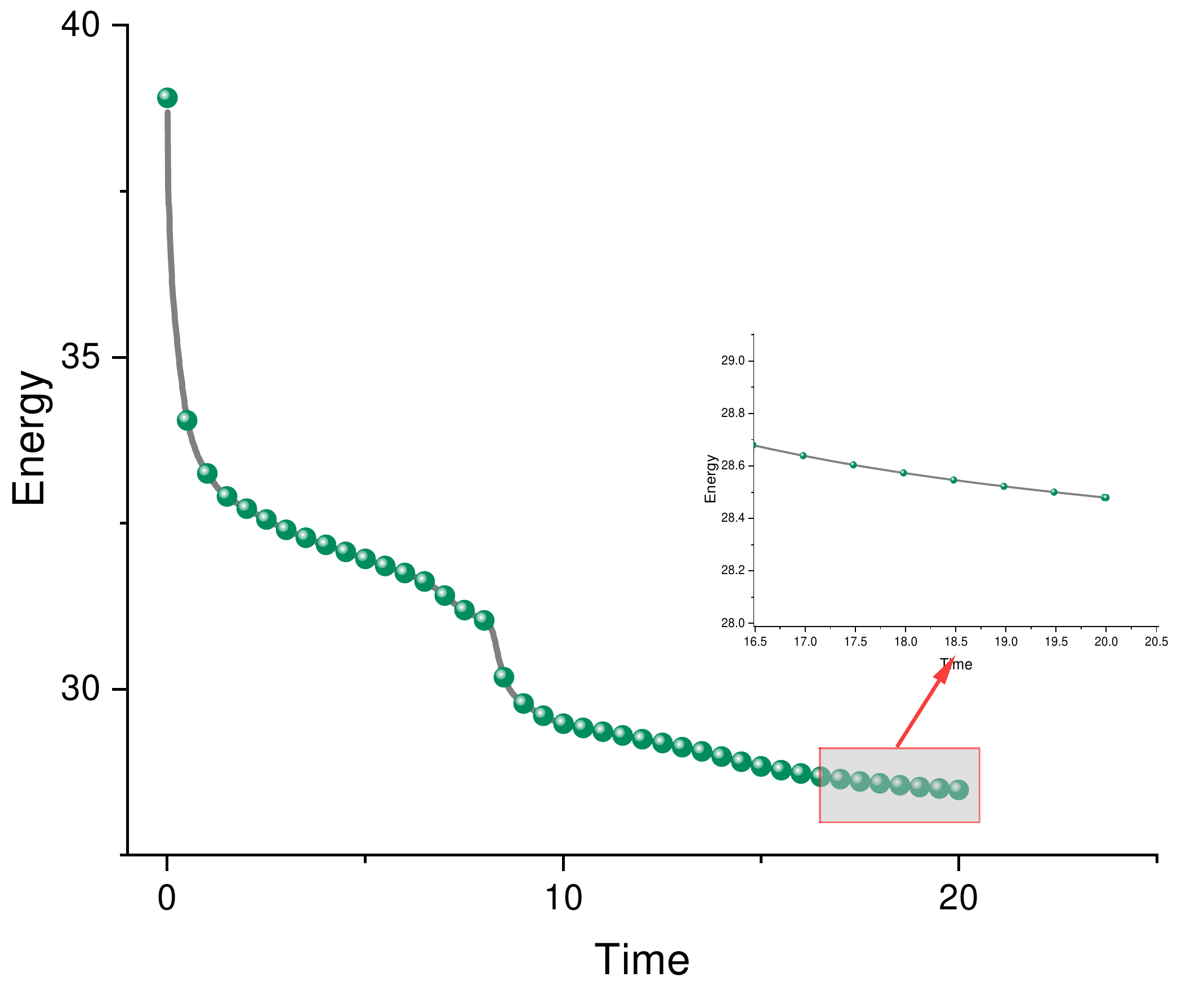}
    \end{minipage}
    }
    \centering      \subfigure[Momentum]{
    \begin{minipage}[b]{0.45\textwidth}
    \centering
    \includegraphics[width=0.9\textwidth,height=0.75\textwidth]{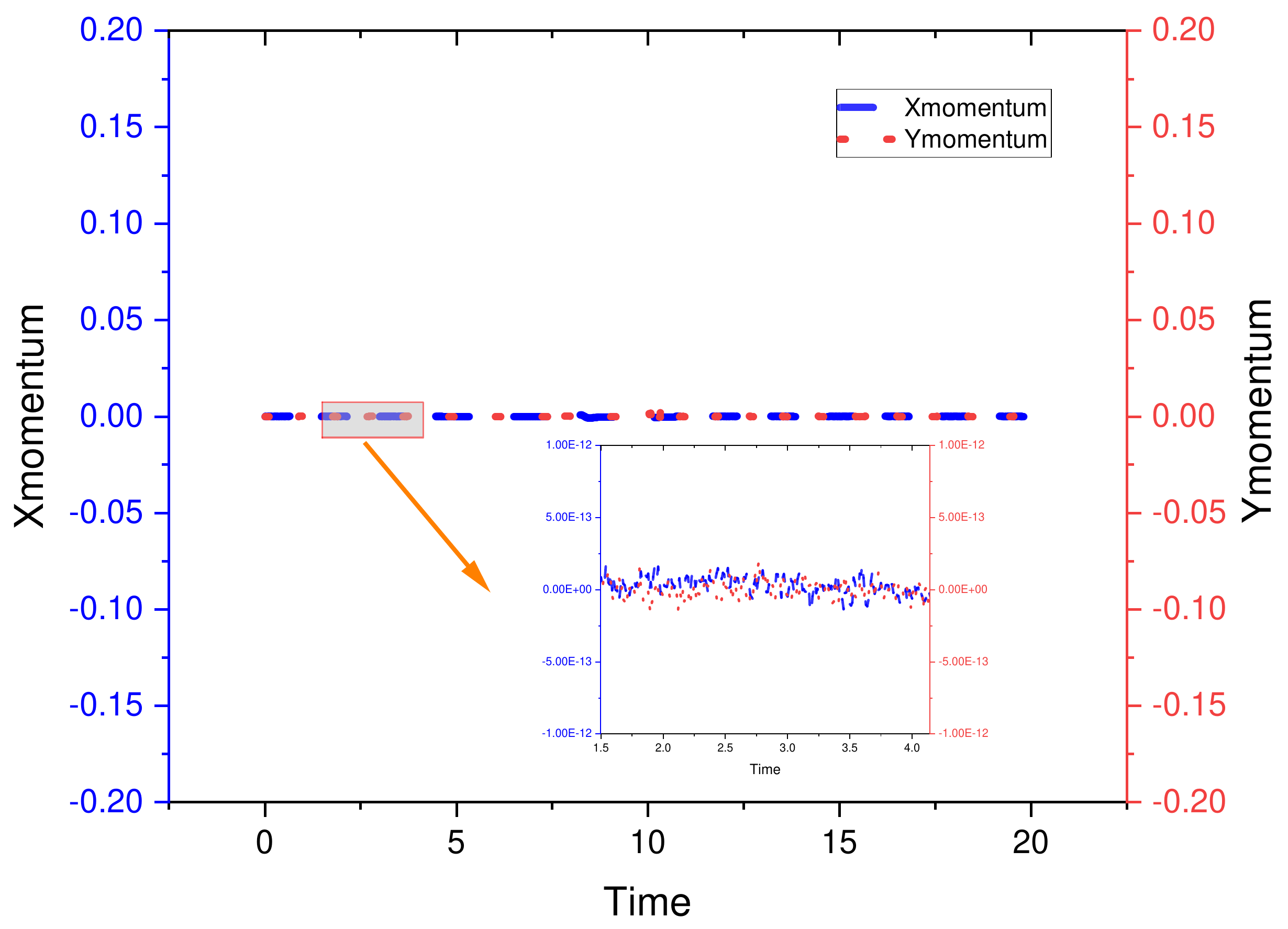}
    \end{minipage}
    }
    \caption{Total energy profile and momentum profile with time} 
    \label{example 1_momentum_energy}
 \end{figure}

In Fig.\ref{example 1_particles}, we observe that the droplet in a square shape deforms into a round one under the interfacial tension. The distribution of particles in the eventual steady state is totally different from the initial uniform distribution. But generally, the spacing between particles remains uniform. So, we conclude that the SPH method can maintain the divergence-free condition well and does not cause tensile instability. The symmetric velocity field and the Fig.\ref{example 1_momentum_energy}(b) also imply momentum conservation at the ODE level and the discrete level. We find the momentum oscillates a little bit numerically around zero from the zoom window. Particles at the interface are also distributed one round layer by one round layer, which is caused by the interfacial force defined by the gradient of $\phi$. While in the bulk phase part, the particles are located more randomly and messily, which makes sense because of the homogeneity.

 In order to show the diffuse characteristics of the NSCH model, an initial condition with a sharp interface is set here. The results show the forming of a diffuse interface with limited width as the numerical calculation process progresses. Finally, the ``color'' of particles on the diffuse interface region also obeys a gradual transition from $-1$ to $1$. This phenomenon coincides with the physical principles. Furthermore, the contour of $\phi$ at the continuous domain interpolated by the SPH in Fig.\ref{example 1_contour} shows that it agrees well with the results of the traditional method based on the fine mesh. 

Moreover, we observe the consistent decay of the total energy in the Fig.\ref{example 1_momentum_energy}(a), which validates our design of ODE expression and the fully discrete scheme. From the zoom window, we can see total energy is dissipating with time, even at the end of the stage of numerical testing. Gradually, when this droplet becomes one in a round shape, the system reaches a steady state with the minimum energy.

\subsection{The deformation and rotation of a droplet in the force field}
We consider the case of a single droplet's deformation and rotation phenomena in the square domain. At the initial stage, a force field $g(x,y,t)$ is enforced in this periodic domain to make this droplet rotate and deform, 
$$
g(x,y,t) =\left\{
       \begin{array}{l}
        10\sin(\pi(x+1)), \quad  0\le t \le 0.5, \\
        0, \quad \text{otherwise}.
       \end{array}
       \right.
$$
The parameters are chosen as: $h = 0.05,\,\Delta t = 0.01,\, T = 20, \,\mathcal{M} = 10^3, \, \eta = 1 , \lambda = 1\, ,\, \epsilon = 0.02,\,\text{Mobility} = 0.002, \xi = 1.$
Initially, the diffusive interface is smoothed by $\phi_0(x,y) = -\tanh((x^2+y^2-0.3^2)/0.02)$
 %
\begin{figure}[h]
    \centering
    \includegraphics[width=0.4\textwidth,height=0.33\textwidth]{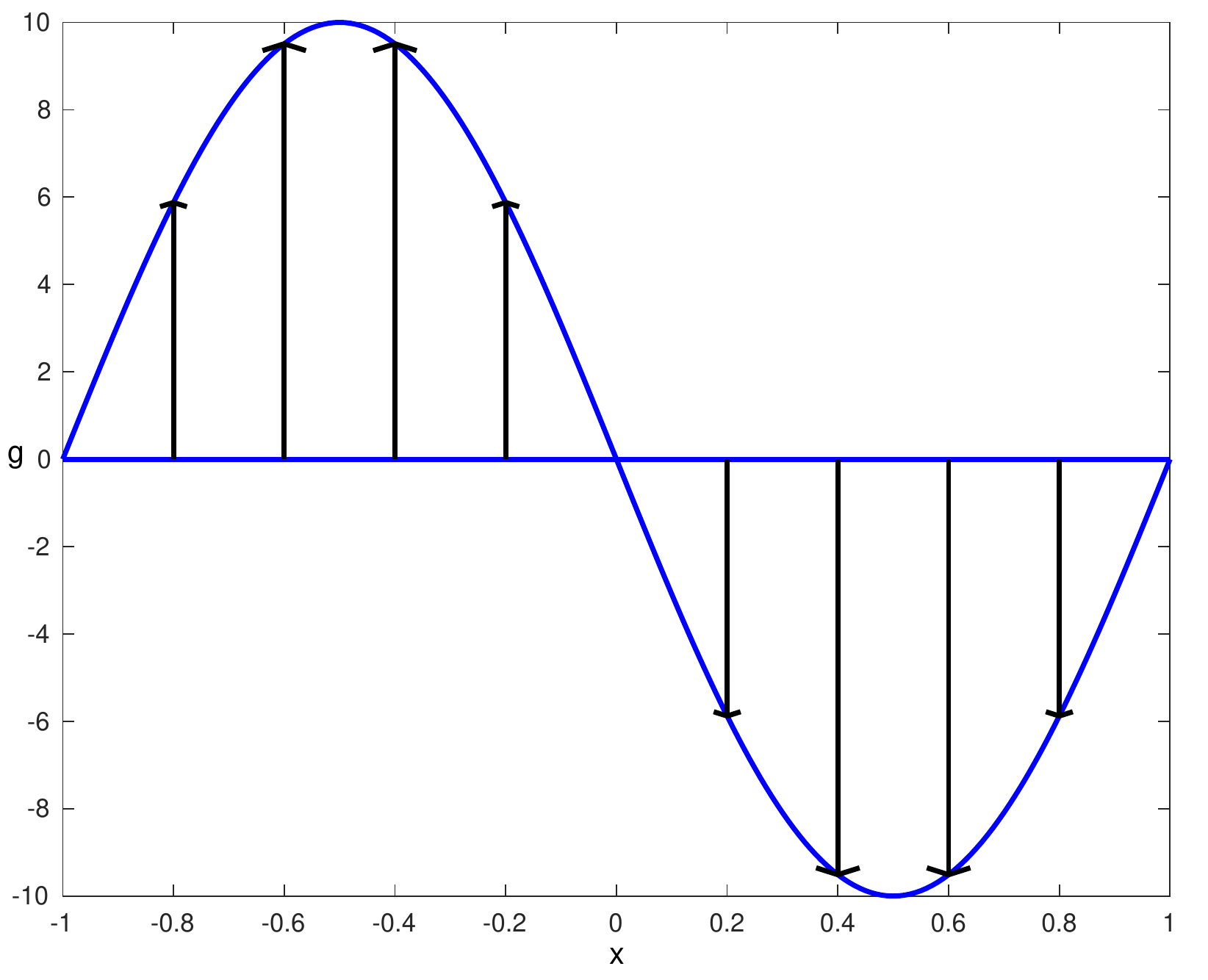}
    \caption{The diagram of force field}
    \label{example 2_force}
 \end{figure}
 \begin{figure}[!b]
    \centering \subfigure[$t=0.01$]{
    \begin{minipage}[b]{0.20\textwidth}
    \centering
    \includegraphics[width=1.0\textwidth,height=0.9in]{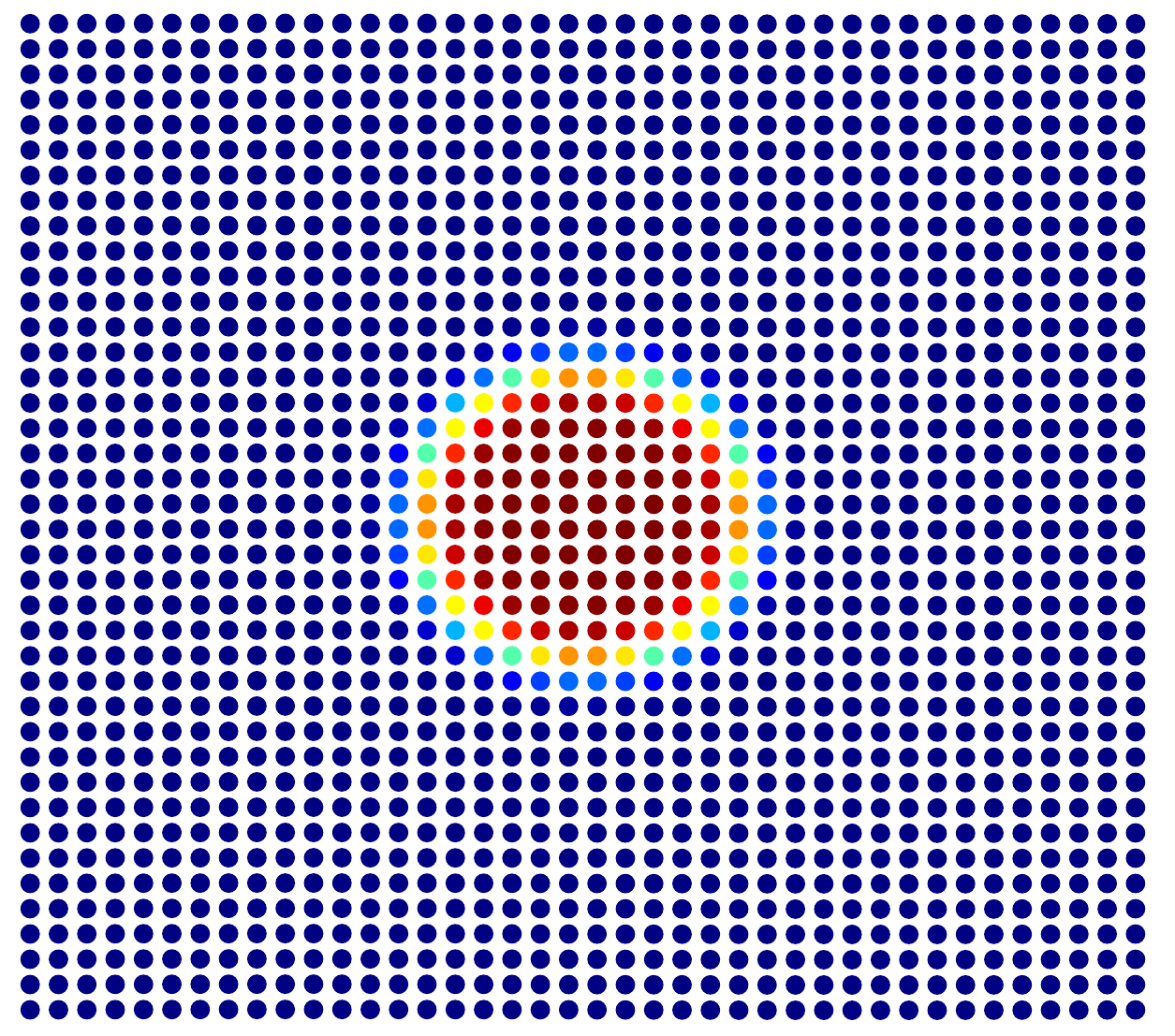}
    \end{minipage}
    }
    \centering \subfigure[$t=0.25$]{
    \begin{minipage}[b]{0.20\textwidth}
    \centering
    \includegraphics[width=1.0\textwidth,height=0.9in]{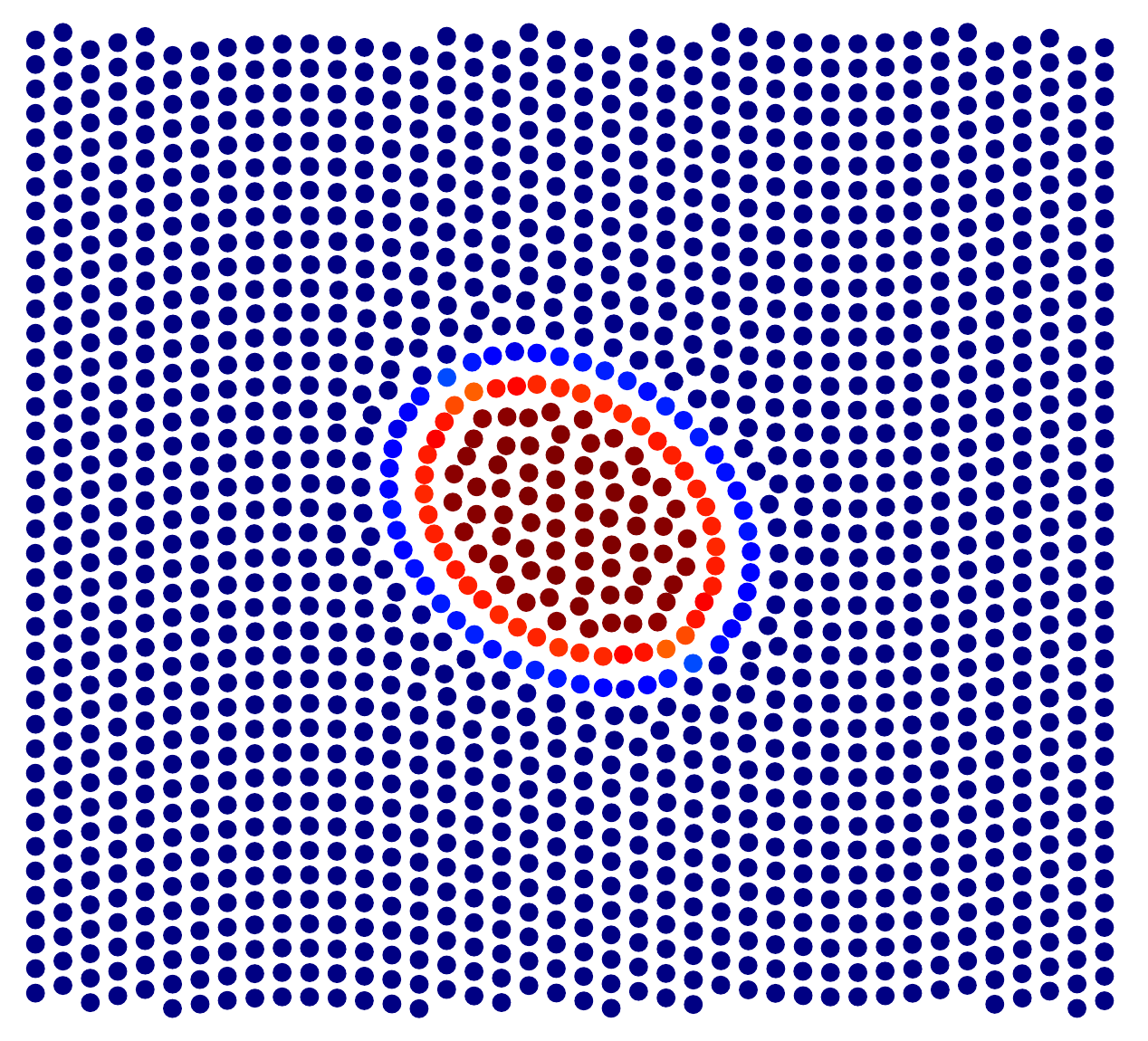}
    \end{minipage}
    }
    \centering \subfigure[$t=0.5$]{
    \begin{minipage}[b]{0.20\textwidth}
    \centering
    \includegraphics[width=1.0\textwidth,height=0.9in]{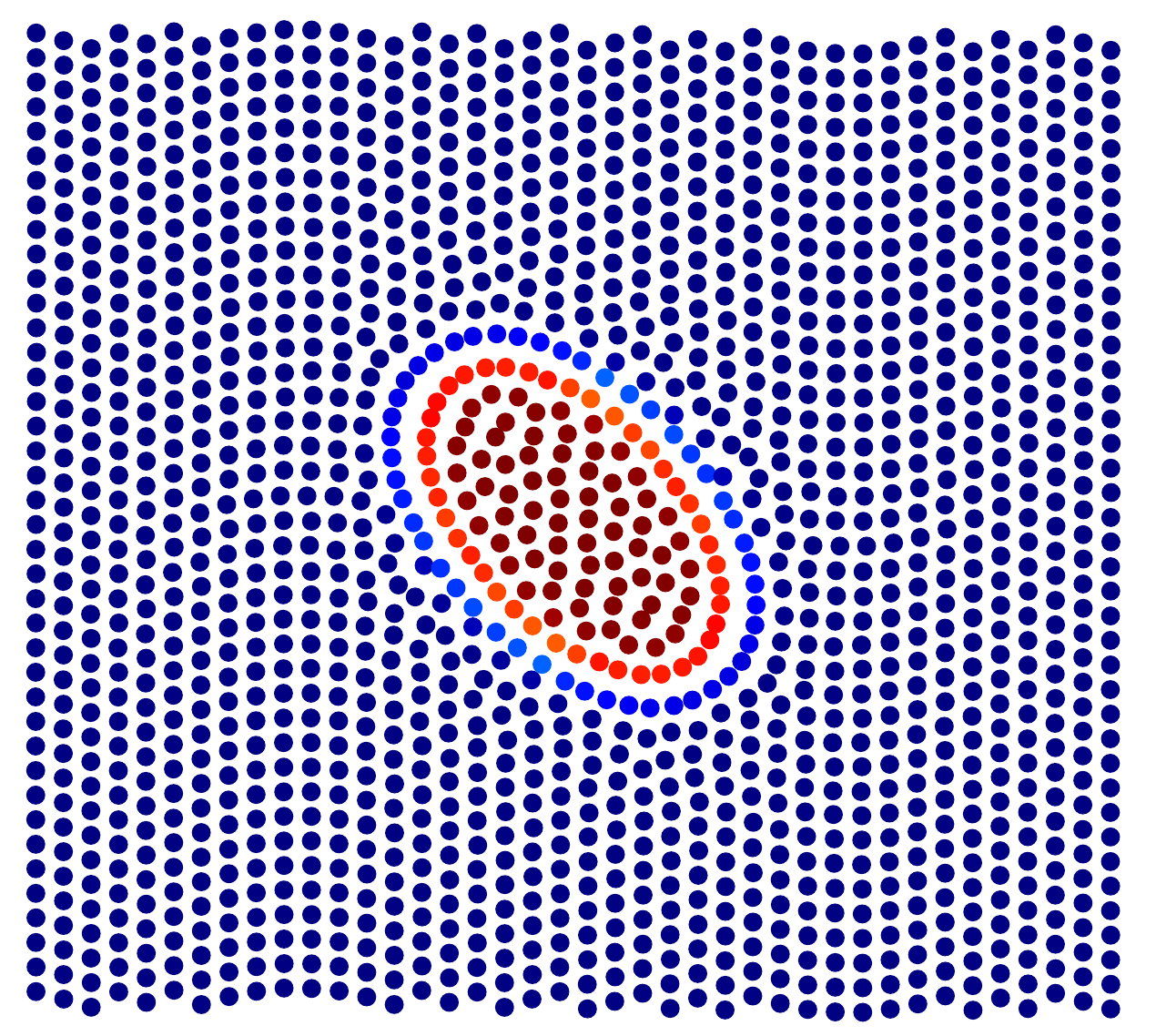}
    \end{minipage}
    }
    \centering \subfigure[$t=20$]{
    \begin{minipage}[b]{0.20\textwidth}
    \centering
    \includegraphics[width=0.95\textwidth,height=0.9in]{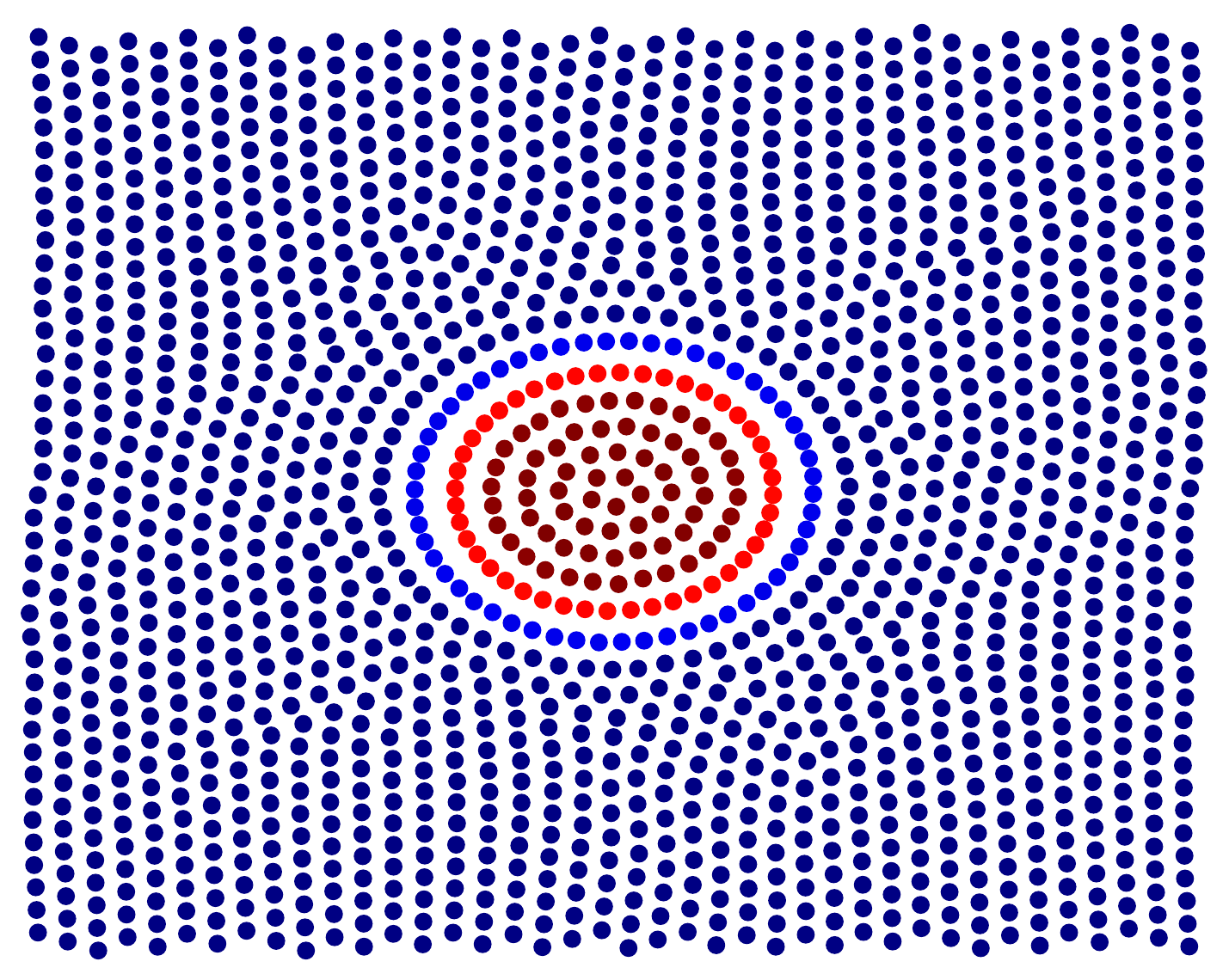}
    \end{minipage}
    }
    \caption{The evolution of $\phi$ and spatial particle distribution with time in example 2}
    \label{example 2_particles}
 \end{figure}
 \begin{figure}[!b]
    \centering \subfigure[$t=0.01$]{
    \begin{minipage}[b]{0.20\textwidth}
    \centering
    \includegraphics[width=0.95\textwidth,height=0.9in]{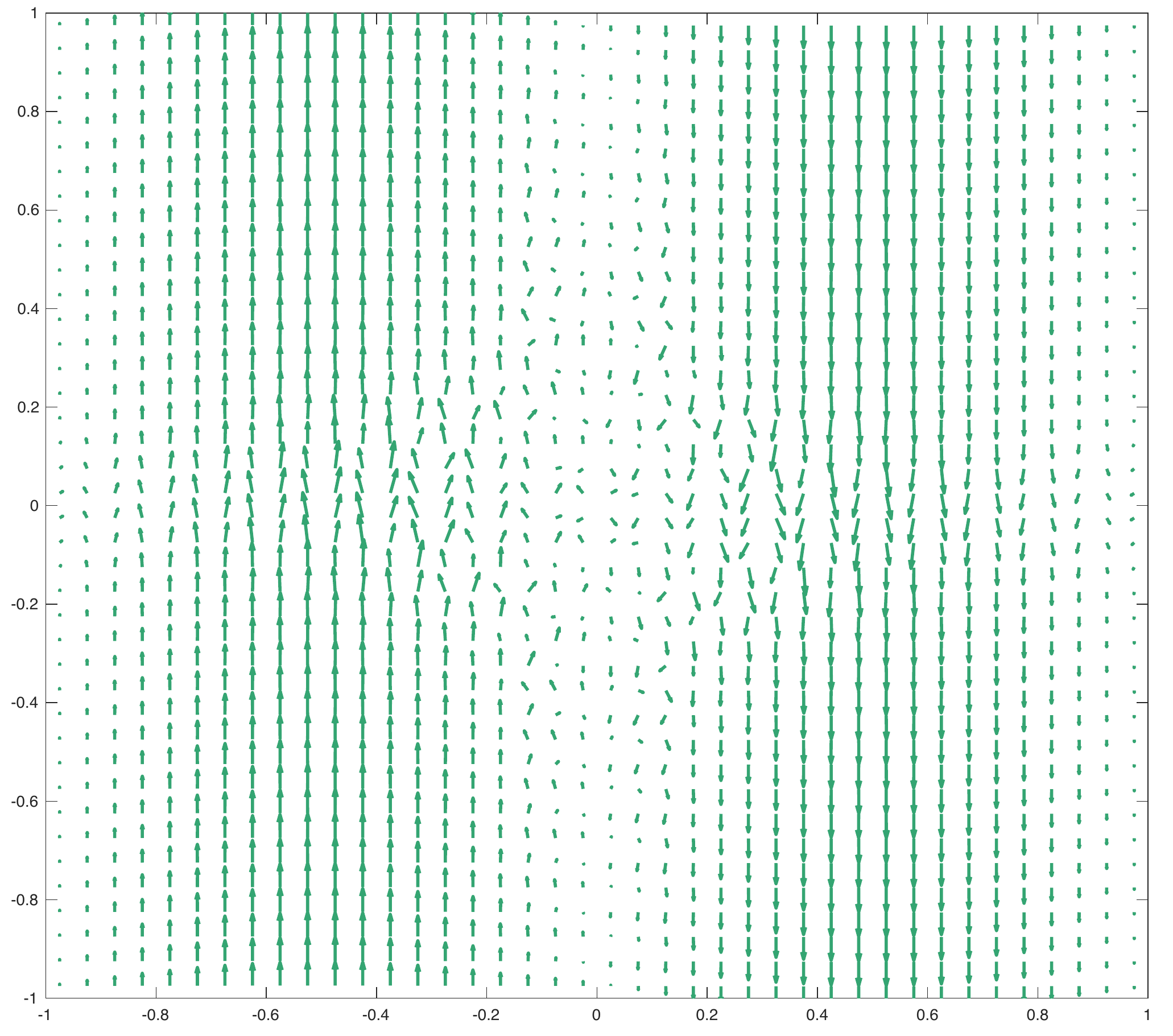}
    \end{minipage}
    }
    \centering \subfigure[$t=0.25$]{
    \begin{minipage}[b]{0.20\textwidth}
    \centering
    \includegraphics[width=0.95\textwidth,height=0.9in]{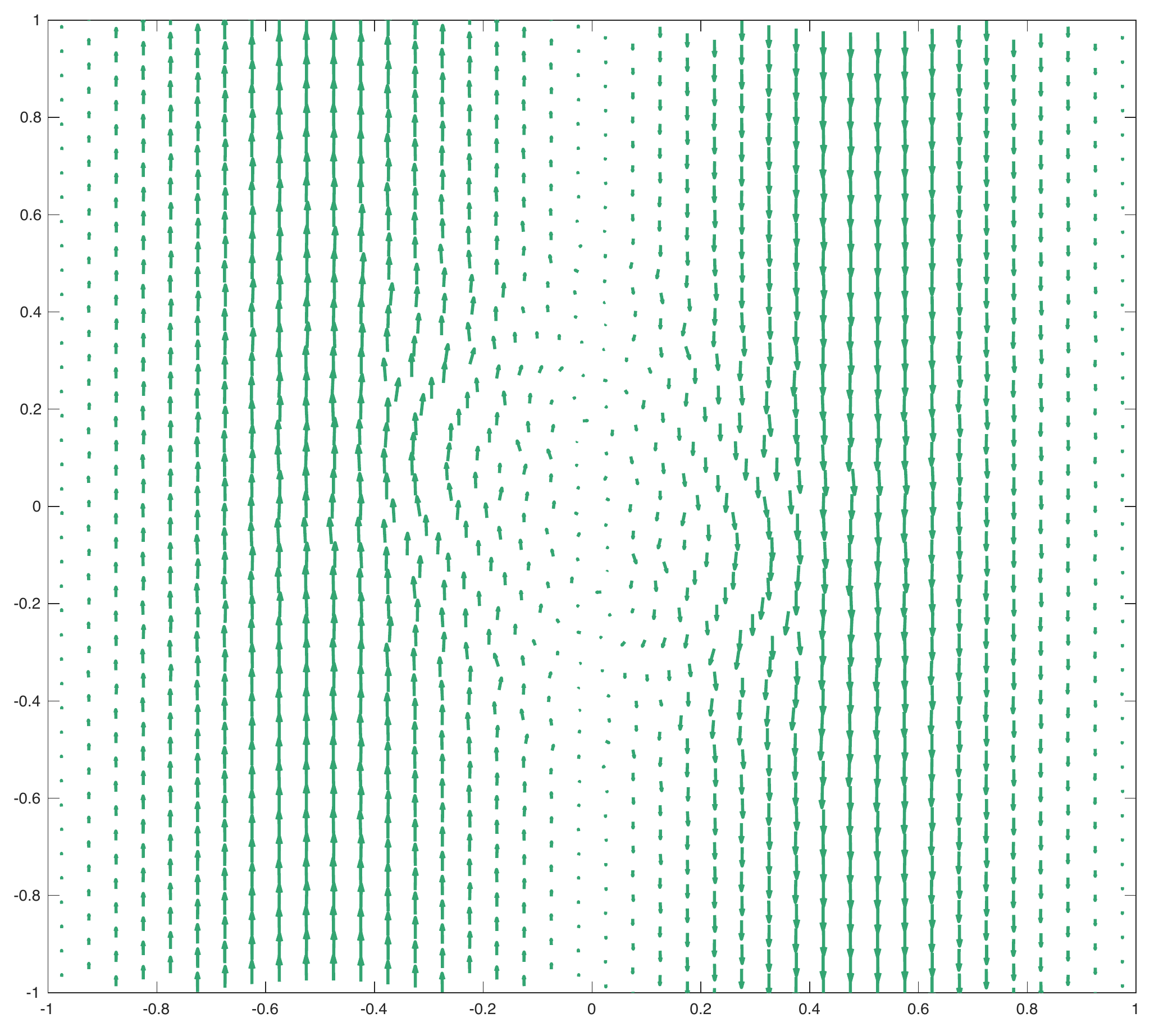}
    \end{minipage}
    }
    \centering \subfigure[$t=0.5$]{
    \begin{minipage}[b]{0.20\textwidth}
    \centering
    \includegraphics[width=0.95\textwidth,height=0.9in]{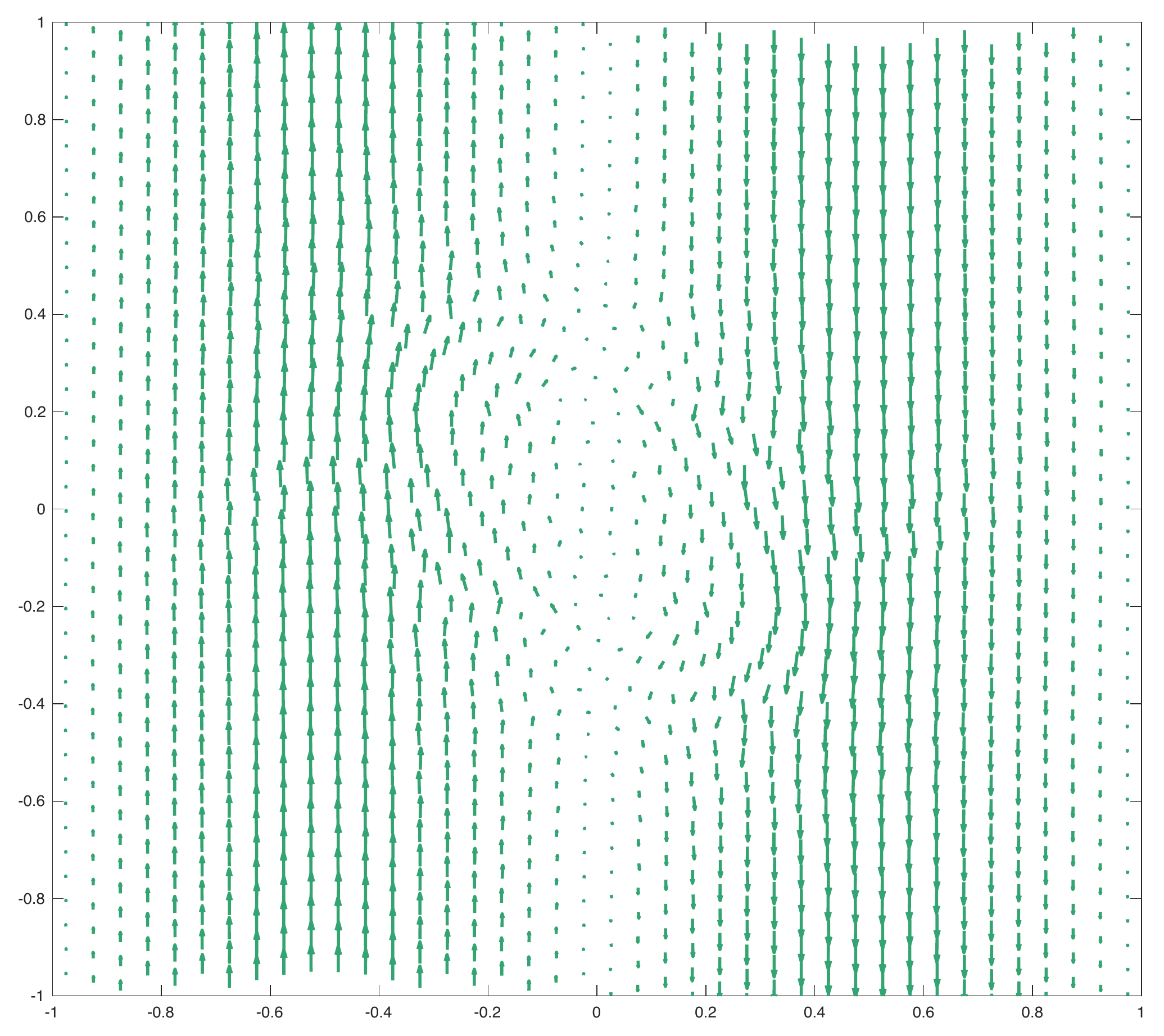}
    \end{minipage}
    }
    \centering \subfigure[$t=20$]{
    \begin{minipage}[b]{0.20\textwidth}
    \centering
    \includegraphics[width=0.95\textwidth,height=0.9in]{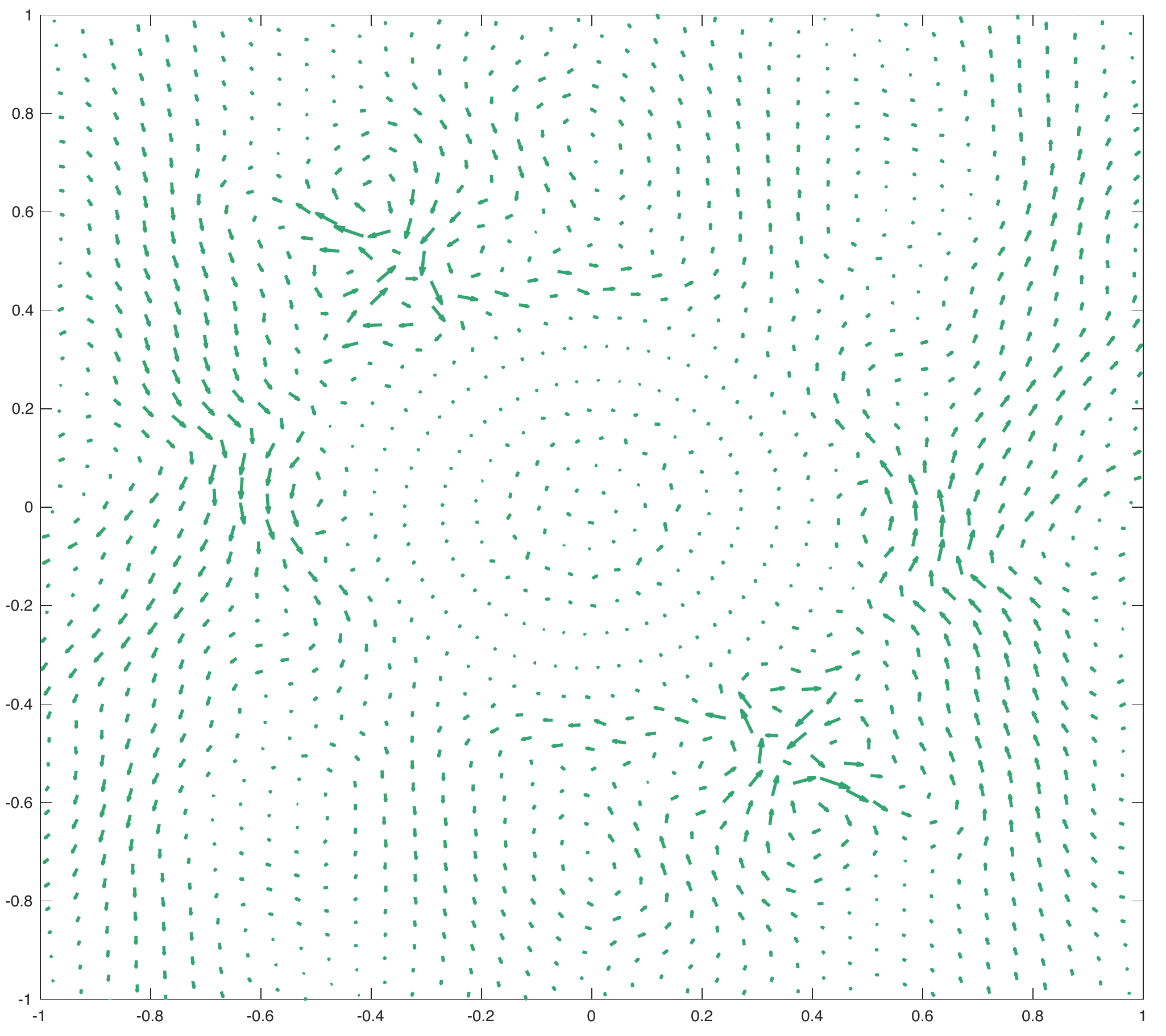}
    \end{minipage}
    }
    \caption{The evolution of velocity quiver plot}
    \label{example2_vel}
 \end{figure}
 \begin{figure}[!t]
    \centering \subfigure[Energy]{
    \begin{minipage}[b]{0.45\textwidth}
    \centering
    \includegraphics[width=0.9\textwidth,height=0.75\textwidth]{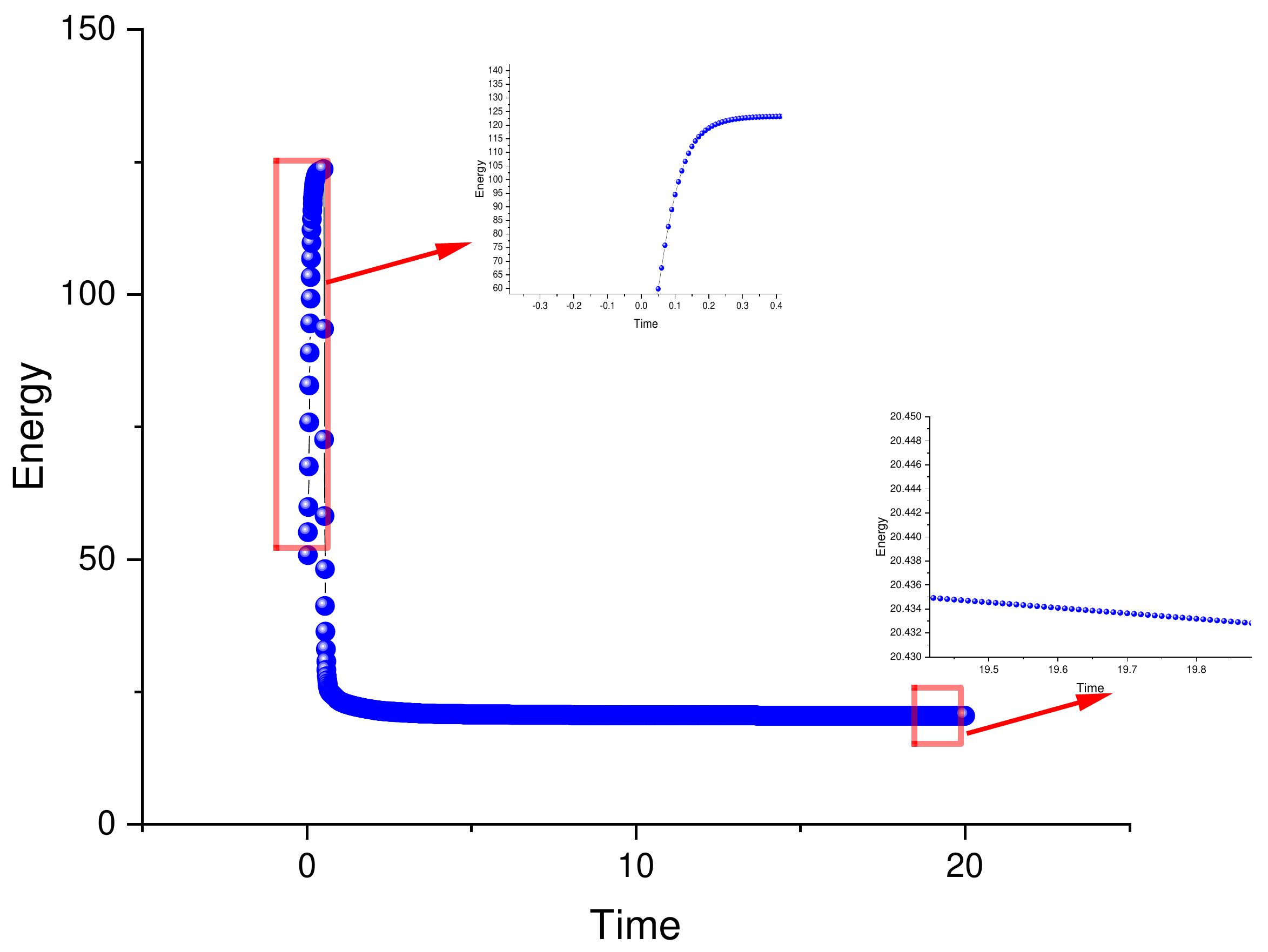}
    \end{minipage}
    }
    \centering      \subfigure[Momentum]{
    \begin{minipage}[b]{0.45\textwidth}
    \centering
    \includegraphics[width=0.9\textwidth,height=0.75\textwidth]{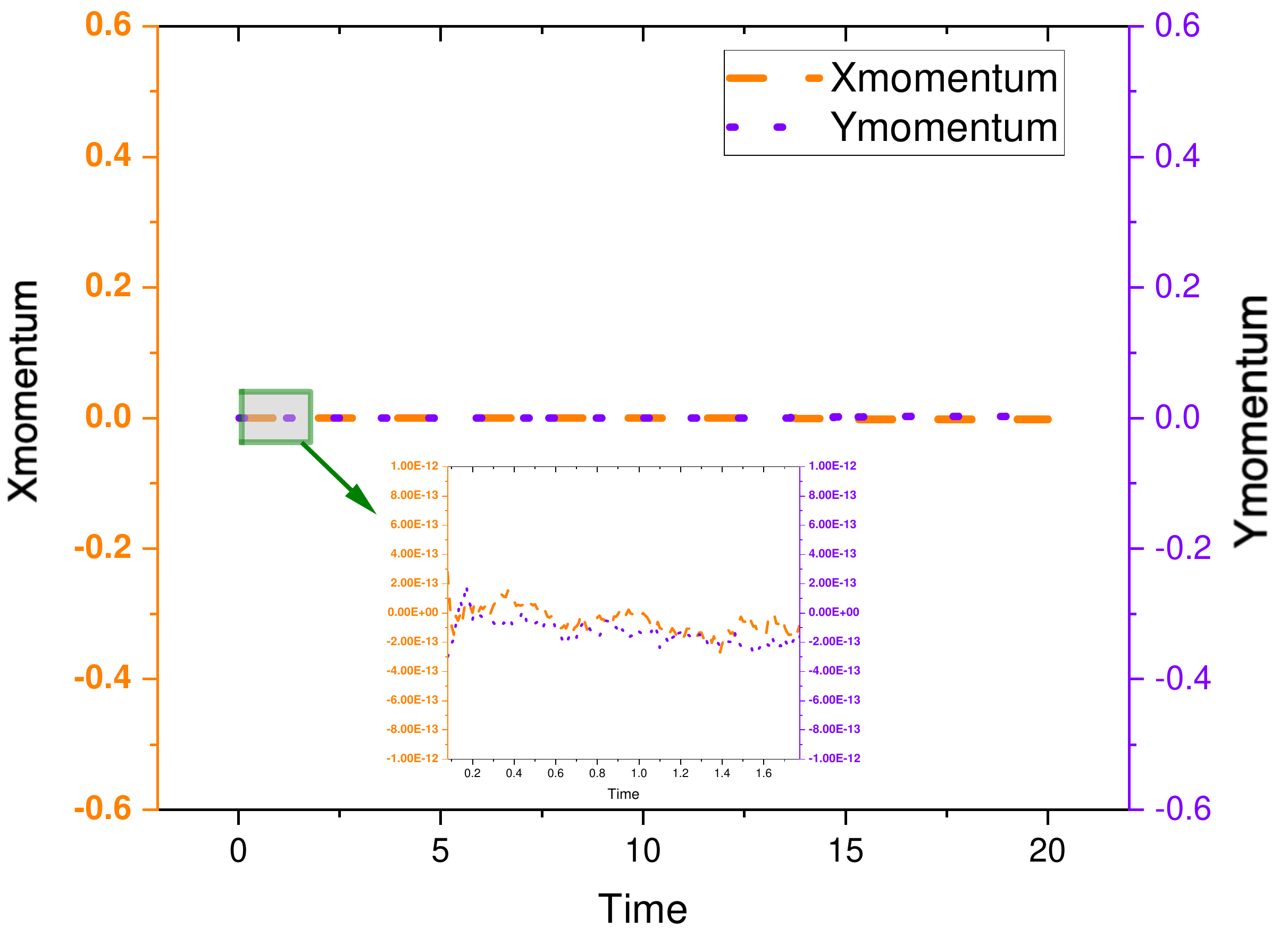}
    \end{minipage}
    }
    \caption{Total energy profile and momentum profile with time}     \label{example 2_energy_momentum}
 \end{figure}

The diffuse interface we set up artificially as the initial conditional gradually evolves into an interface with more physical transition layers. This example shows the good adaptivity to the two-phase flow problems with high deformation. The interface can be tracked naturally with the $\phi$ information carried by the particles as Fig.\ref{example 2_particles}. The symmetric velocity quiver plots as Fig.\ref{example2_vel} and Fig.\ref{example 2_energy_momentum}(b) indicate the maintenance of divergence-free conditions and momentum conservation.

In Fig.\ref{example 2_energy_momentum}(a), the total energy of the system is increasing at the beginning, in the time period $0–0.5$, due to the external work done by the force field. As we expected, after cancelling this force field at time = $0.5$, the total energy of the system decays with time due to the viscous dissipation and the diffusive effect of the CH system and the momentum conservation holds well even facing relative high flow rate at the initial stage as well.
 %
 \subsection{Droplets of different sizes merge into one}
 In this example, we consider the case of four circular droplets merging with one large droplet in the computational domain. The parameters are chosen as:
$h = 0.05, \Delta t = 0.01, T = 10,  \mathcal{M} = 10^3,  \eta = 1 , \lambda = 1, \epsilon = 0.02, \text{Mobility} = 0.002, \xi=1.$
The initial phase variable is chosen as:
$$
\begin{aligned}
\phi_0(x,y) &= -\tanh((x^2+y^2-0.3^2)/0.01)* \\
    &\tanh(((x-0.4)^2+y^2-0.1^2)/0.01)*
     \tanh(((x-0.2)^2+y^2-0.1^2)/0.01)*
     \\
    &\tanh(((y-0.2)^2+x^2-0.1^2)/0.01)*
     \tanh(((y-0.4)^2+x^2-0.1^2)/0.01).
\end{aligned}
$$

The free energy contours in Fig.\ref{example 3_free} show that the maximum free energy decreases from $1.5$ to $0.3$. At the initial stage, the interfaces of droplets with different sizes connect with each other with very sharp angles. Since the free mixing energy characterizes the surface intension, the sharp connection naturally leads to a higher free energy in this region. The maximum value of the free energy gradually decreases as the angle becomes round and smooth. The homogenization of $\phi$ at the interface region also helps decrease the maximum value of free energy. The energy dissipation and momentum conservation are preserved well in example 3 as Fig.\ref{example 3_energy_momentum}.
\begin{figure}[!htb]
    \centering \subfigure[Initial]{
    \begin{minipage}[b]{0.3\textwidth}
    \centering
    \includegraphics[width=1.0\textwidth,height=1.4in]{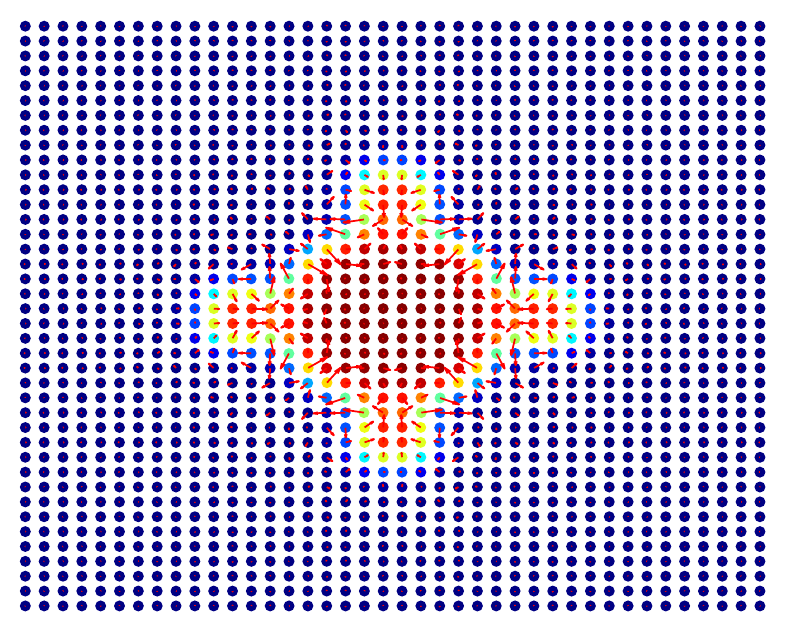}
    \end{minipage}
    }
    \centering \subfigure[$t=1$]{
    \begin{minipage}[b]{0.3\textwidth}
    \centering
    \includegraphics[width=1.0\textwidth,height=1.4in]{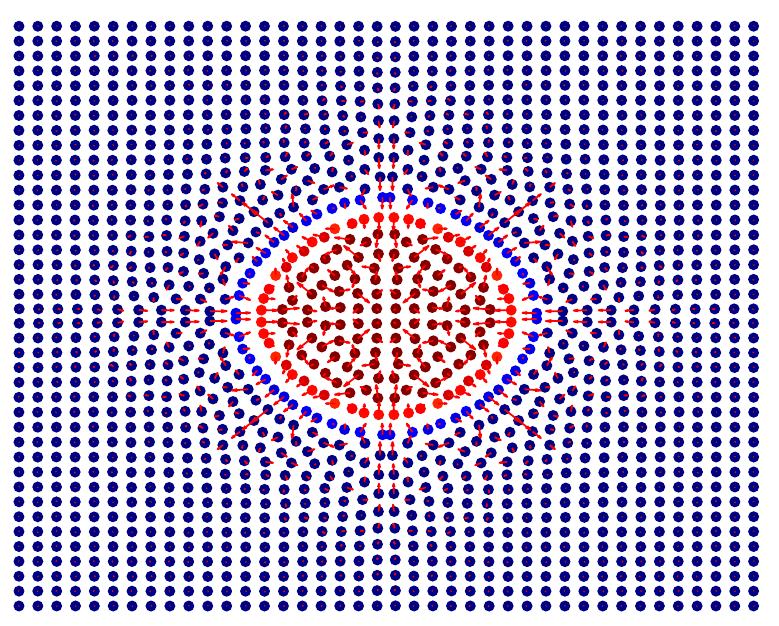}
    \end{minipage}
    }
    \centering \subfigure[$t=30$]{
    \begin{minipage}[b]{0.3\textwidth}
    \centering
    \includegraphics[width=1.0\textwidth,height=1.4in]{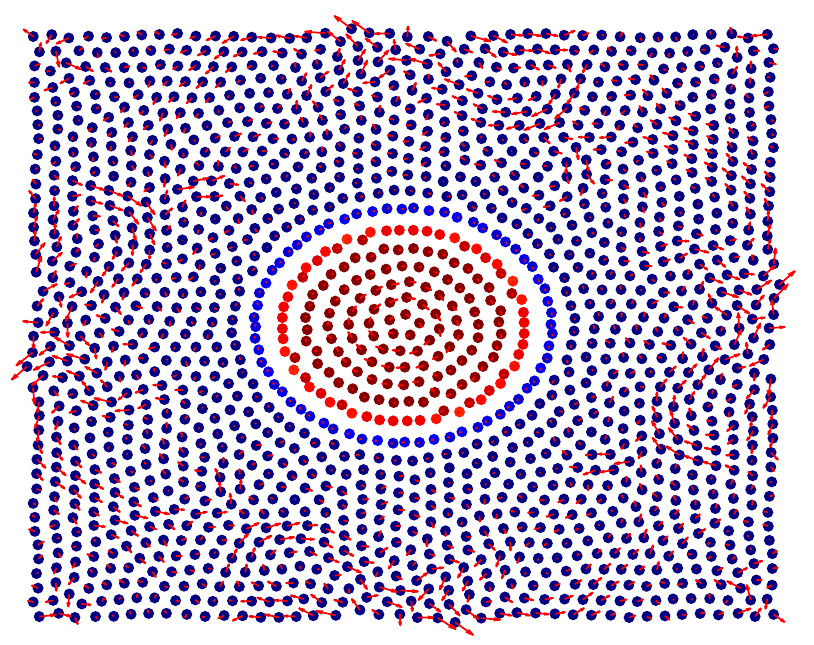}
    \end{minipage}
    }
    \caption{The evolution of $\phi$ and spatial particle distribution with time in example 3}
    \label{example 3_particles}
\end{figure}
\begin{figure}[!htb]
    \centering \subfigure[Initial]{
    \begin{minipage}[b]{0.3\textwidth}
    \centering
    \includegraphics[width=1.0\textwidth,height=1.4in]{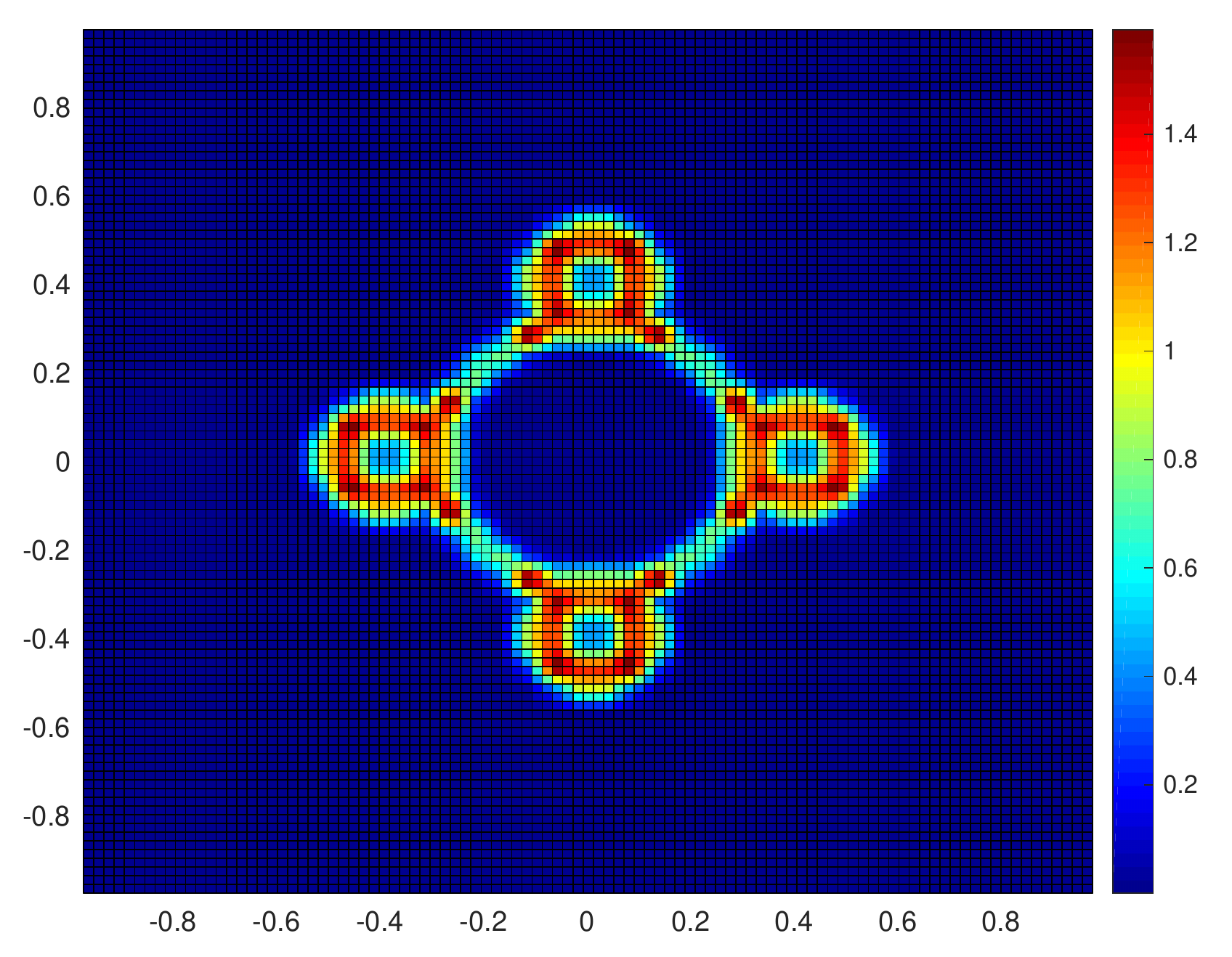}
    \end{minipage}
    }
    \centering \subfigure[$t=1$]{
    \begin{minipage}[b]{0.3\textwidth}
    \centering
    \includegraphics[width=1.0\textwidth,height=1.4in]{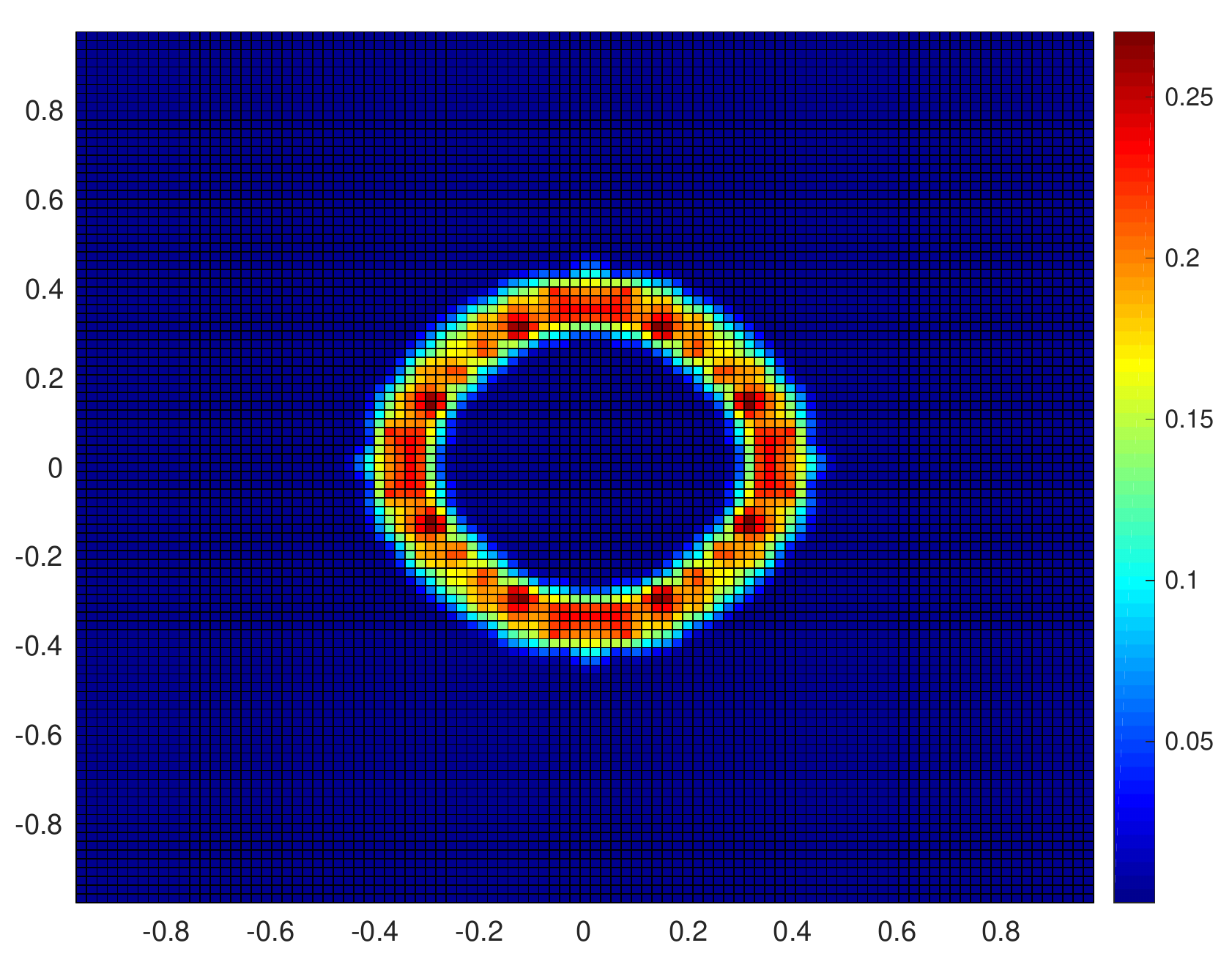}
    \end{minipage}
    }
    \centering \subfigure[$t=30$]{
    \begin{minipage}[b]{0.3\textwidth}
    \centering
    \includegraphics[width=1.0\textwidth,height=1.4in]{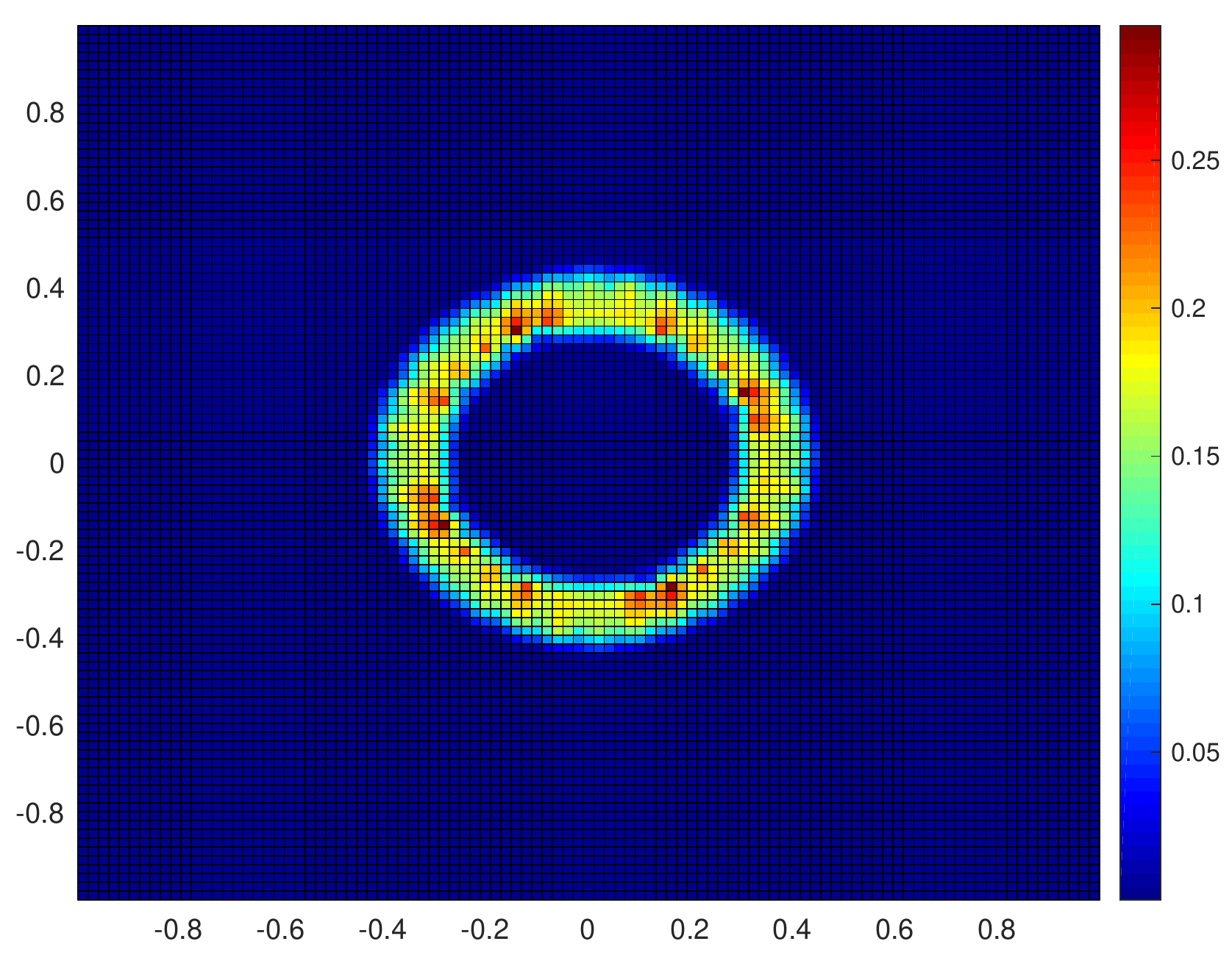}
    \end{minipage}
    }
    \caption{The evolution of free energy throughout the entire domain}
    \label{example 3_free}
\end{figure}
 \begin{figure}[!t]
    \centering \subfigure[Energy]{
    \begin{minipage}[b]{0.45\textwidth}
    \centering
    \includegraphics[width=0.9\textwidth,height=0.75\textwidth]{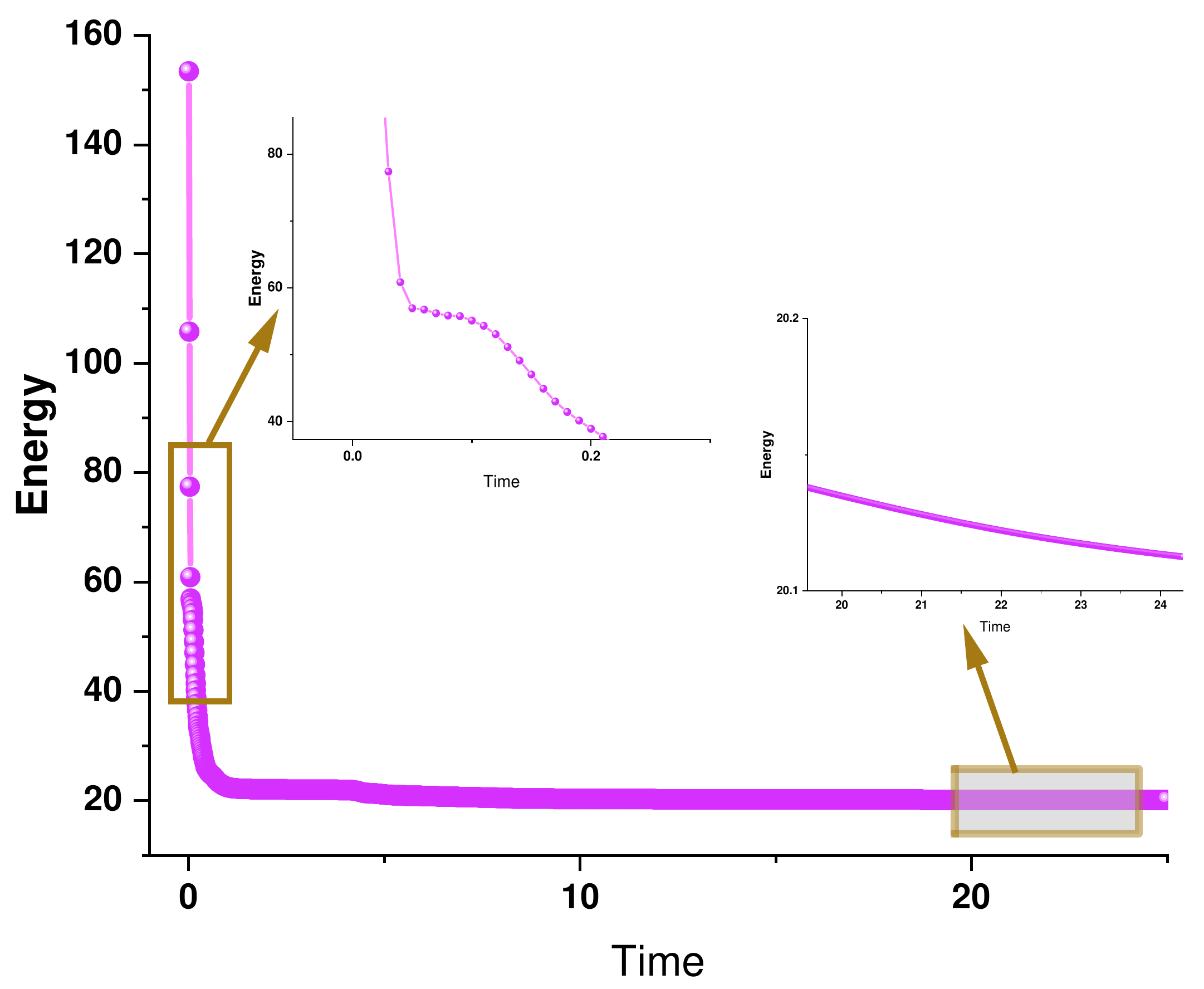}
    \end{minipage}
    }
    \centering      \subfigure[Momentum]{
    \begin{minipage}[b]{0.45\textwidth}
    \centering
    \includegraphics[width=0.9\textwidth,height=0.75\textwidth]{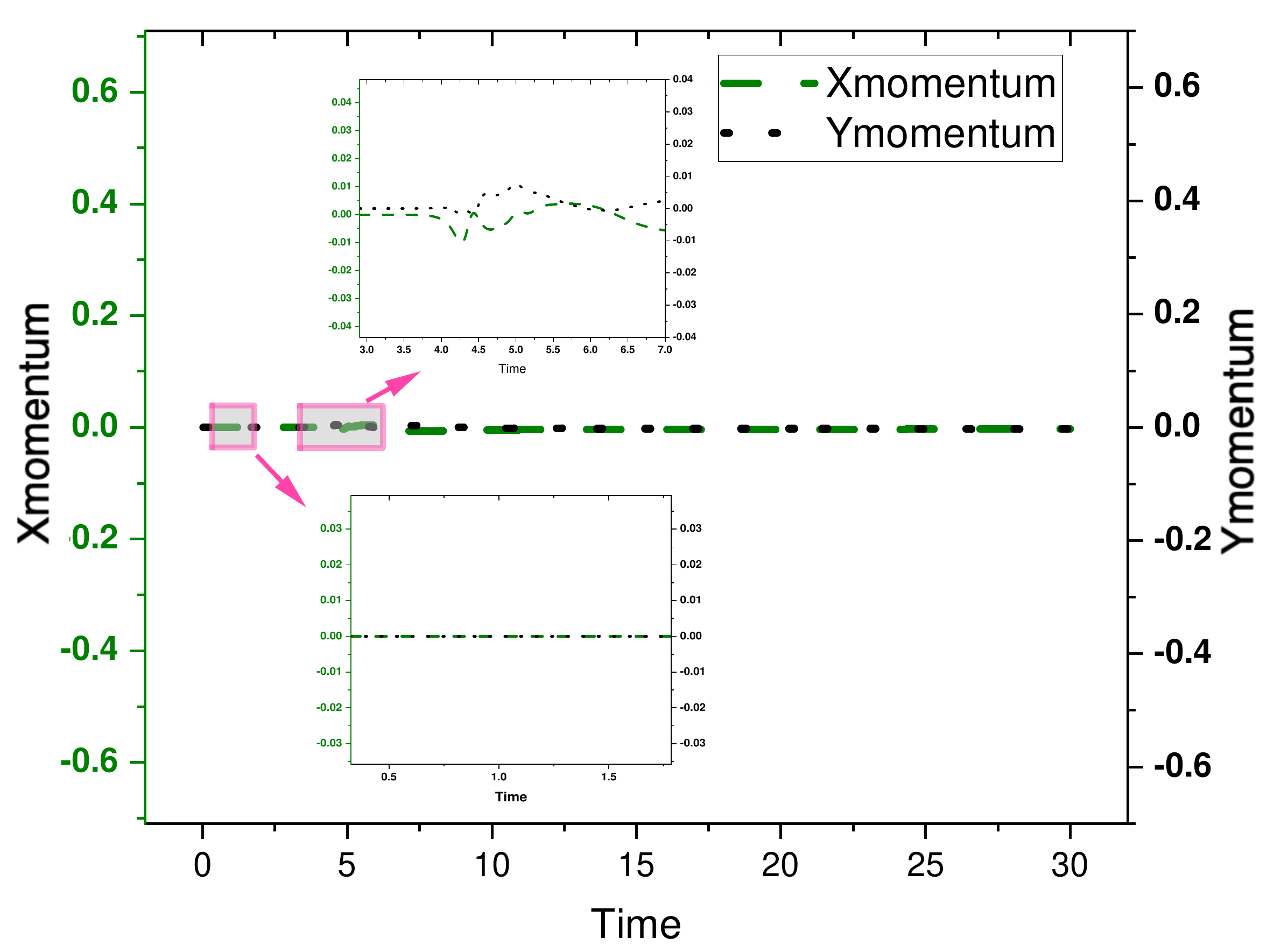}
    \end{minipage}
    }
    \caption{Total energy profile and momentum profile with time}
    \label{example 3_energy_momentum}
 \end{figure}



\section{Conclusion}  \label{section:Conclusion}
In this work, we exploit the great potential of the Lagrangian particle treatment in the SPH, which also fulfills the increasing demand for particle modeling that allows more physical properties to be incorporated. To the best of our knowledge, this is the first study of energy-stable ODE discretization and the energy-stable fully discrete scheme by the SPH method for two-phase problems. Our proposed scheme ensures the inheritance of momentum conservation and the energy dissipation law from the PDE level to the ODE level, and then to the fully discrete level. Consequently and desirably, it also helps increase the stability of the numerical method. The time step size can be much larger than that of the traditional ISPH methods. This energy-stable SPH method also alleviates the tensile instability without using any particle-shifting strategies, which may destroy the rigorous mathematical proof. The numerical results also demonstrate that our method captures the interface behavior and the energy variation process well.

\section*{Acknowledgments}
This work is supported by King Abdullah University of Science and Technology (KAUST) through the grants BAS/1/1351-01, URF/1/407401, and URF/1/3769-01. Z. Qiao's work is partially supported by the Hong Kong Research Grants Council (RFS Project No.RFS2021-5S03 and GRF project No.15302919) and the Hong Kong Polytechnic University internal grant No.4-ZZLS.

\appendix

\bibliographystyle{elsarticle-num} 
\bibliography{reference}
\end{document}